\documentclass[a4paper]{article}
\usepackage{latexsym}
\usepackage{amsmath,amssymb,mathrsfs,pictex} 
\usepackage{amsmath,amsthm,amsfonts,amscd,eucal}
\usepackage{graphicx}
\usepackage{epsfig}
\usepackage{epstopdf} 

\usepackage[affil-it]{authblk}

\numberwithin{equation}{section}

\def\cb{{\mathcal B}}

\def\cd{{\mathcal D}}

\def\cf{{\mathcal F}}

\def\ch{{\mathcal H}}

\def\cs{{\mathcal S}}

\def\bc{{\mathbb C}}

\def\bn{{\mathbb N}}

\def\br{{\mathbb R}}
\def\bs{{\mathbb S}}

\def\bz{{\mathbb Z}}

\def\ga{{\mathfrak A}}

 \def\gph{{\mathfrak h}}

\def\a{\alpha}
\def\b{\beta}
\def\g{\gamma}        \def\G{\Gamma}
\def\d{\delta}        \def\D{\Delta}
\def\eps{\varepsilon}

\def\l{\lambda}       \def\La{\Lambda}
\def\m{\mu}
\def\n{\nu}

\def\r{\rho}
\def\s{\sigma}

\def\f{\varphi} 

\def\om{\omega}        \def\Om{\Omega}

\newtheorem{Thm}{Theorem}[section]

\newtheorem{Prop}[Thm]{Proposition}
\newtheorem{Lemma}[Thm]{Lemma}

\theoremstyle{definition}
\newtheorem{Dfn}[Thm]{Definition}

\theoremstyle{remark}

\def\di{\mathop{\mathop{\rm d}}\!}

\def\im{\mathop{\rm Im}}

\newcommand{\svv}{\text{span}\,}

\newcommand{\carf}{\textrm{CAR}}
\newcommand{\ccr}{\textrm{CCR}}

\newcommand{\tr}{\textrm{Tr}}

\def\idd{{\bf 1}\!\!{\rm I}}

\def\supp{\mathop{\rm supp}}

\newcommand{\nn}{\nonumber}

\DeclareMathOperator{\vol}{vol}

\begin{document}

\title{Bose-Einstein condensation and condensation of $q$-particles in equilibrium and non equilibrium thermodynamics: a new approach}
\author{Luigi Accardi%
\thanks{Electronic address: \texttt{accardi@volterra.uniroma2.it}}}
\affil{Centro Vito Volterra, Universit\`{a} di Roma Tor Vergata, 
Via Columbia 2, Roma 00133, Italy}
\author{Francesco Fidaleo%
\thanks{Electronic address: \texttt{fidaleo@mat.uniroma2.it}; corresponding author}}
\affil{Dipartimento di Matematica, Universit\`{a} di Roma Tor Vergata, 
Via della Ricerca Scientifica 1, Roma 00133, Italy}



\maketitle

\begin{abstract}

In the setting of the principle of local equilibrium which asserts that the temperature is a function of the energy levels of the system, we exhibit plenty of steady states describing the condensation of free Bosons which are not in thermal equilibrium. The surprising facts are that the condensation can occur both in dimension less than 3 in configuration space, and even in excited energy levels. The investigation relative to non equilibrium suggests a new approach to the condensation, which allows an unified analysis involving also the condensation of $q$-particles, $-1\leq q\leq 1$, where $q=\pm1$ corresponds to the Bose/Fermi alternative. For such $q$-particles, the condensation can occur only if $0<q\leq1$, the case 1 corresponding
to the standard Bose-Einstein condensation. In this more general approach, completely new and unexpected states exhibiting condensation phenomena naturally occur also in the usual situation of equilibrium thermodynamics. 
The new approach proposed in the present paper for the 
situation of $2^\text{nd}$ quantisation of free particles, is naturally based on the theory of the Distributions, which might hopefully be extended to more general cases.
\end{abstract}

\section{introduction}

The investigation of the Bose-Einstein Condensation (BEC for short) has a very long history after discovering a new particle statistics by Satyendra Nath Bose (cf. \cite{B})
and Albert Einstein (see e.g. \cite{C}) at the beginning of 20th century. It concerns the fact that a macroscopic amount of 
elementary particles having integer spin (called Bosons in honour of N. Bose) can occupy the ground state after the thermodynamic limit. The simplest and most important ideal model is that consisting of free massive Bosons. We also mention the Gross-Pitaevskii equation (cf. \cite{G, P}) describing a weakly interacting Bose-Einstein condensate at low temperature. Many phenomena involving quasi-particles like phonons or magnons can be described by the BEC. Recently, in \cite{KSVW} the condensation of massless particles like photons has been pointed out as well. Also the phenomenon of the superfluidity of the helium isotope $He_4$ seems to be tightly but not directly connected with the BEC. In fact, on one hand the superfluidity is mentioned as one of the most important examples confirming the evidence of the condensation effects. On the other hand, it is unclear if this phenomenon is directly connected with the condensation of (quasi) particles, and was explained by Lev Landau (cf. \cite{L}) with a fluid composed of two indivisible components, one in the superfluid phase and the other one in the normal phase. At zero temperature, the $He_4$ is entirely composed by the superfluid component. The thermal agitation causes the excitation of the part of the dispersion spectrum corresponding to phonons, and mainly to rotons. In such a way, as the temperature increases, 
the portion of the superfluid component becomes lover and lover until to disappear at the critical temperature, known as the $\l$-point.

The behaviour at low temperatures of particles obeying to Fermi-Dirac statistics is completely different. Due to Pauli exclusion principle,
it is well known that Fermions (i.e. quantum particles of half-integer spin) does not lead to any condensation by the Pauli exclusion principle. Nevertheless, the other isotope helium $He_3$ still exhibits superfluidity at a temperature very close to $0^o$ Kelvin, even if these are Fermi particles. The superfluidity of $He_3$ can be justified with the fact that at very low temperature, $He_3$-particles form the so called sea of pair-particles which can be considered as Bosons. Namely, also $He_3$ can exhibit superfluidity. According to the BCS theory (cf. \cite{BC}), it is precisely the same phenomenon occurring in superconductors where pairs of electrons forms the so called Bardeen-Cooper pairs. The reader is also referred to \cite{F2, F3, FGI}, where it is shown the condensation of Bardeen-Cooper pairs in arrays of Josephson junctions. In this sense, the phenomenon of superconductivity is also connected with the BEC. Finally, in Sections 11.7 and 11.10 of the 
monograph \cite{ALV}, the BCS pairing appears as a general phenomenon of {\it Bosonisation} in the setting of stochastic limit.
For a exhaustive treatment of the various condensation phenomena and perspectives, we mention the sample very far to be complete, of books \cite{Le, Lo, U}, and \cite{BR} with the literature cited therein, for a rigorous mathematical approach to BEC.

Typically BEC is considered in states which are in equilibrium at a 
certain temperature and with respect to the natural dynamics generated by a given Hamiltonian.
The definition of equilibrium states in terms of Kubo-Martin-Schwinger (KMS for short) 
condition (cf. \cite{K, MS}) w.r.t. to a given dynamics has the advantage 
compared with the usual Gibbs prescription given in terms of the 
Hamiltonian implementing the dynamics, of being valid also after the 
infinite volume limit, see e.g. \cite{HHW}. 
It is known that an equilibrium state for a given dynamics is automatically 
stationary for this dynamics. 

Non Equilibrium Steady States (NESS for short) play also an important role in
non equilibrium thermodynamics. 
A natural class of NESS naturally emerges from the stochastic 
limit of quantum theory, see e.g. \cite{AI, AIK, AKP, ALV} 
and references therein. 
States in this class are called {\it local equilibrium} (or {\it local KMS}
states) where the term {\it local} refers to the fact that
the temperature is a function of the Hamiltonian of the system
so that any energy level can be thought to be in equilibrium at a 
temperature  depending on the level itself. Moreover, these NESS are 
characterised in terms of a {\it Local Equilibrium Principle}, that 
can be considered as a generalisation of the usual KMS equilibrium one.
The Local Equilibrium Principle was formulated in \cite{AI} for systems with pure-point spectrum dynamics. In \cite{AFQ} this principle was 
formulated for systems whose dynamics is implemented by a strongly 
continuous $1$-parameter unitary group acting on a given Hilbert space.
Several equivalent (or almost equivalent) formulations of it were 
also discussed for systems with pure-point spectrum dynamics. 
The present paper extends the above results in several directions as follows.
\vskip0.1cm
\noindent
{\bf (i)} Exploiting the fact that the quasi-free dynamics (i.e. one-parameter group of Bogoliubov automorphisms)
leave the Weyl $C^*$-algebra invariant, we give a 
purely $C^*$-algebraic formulation of the Local Equilibrium Principle for 
quasi-free (gaussian) NESS with respect to a given quasi-free dynamics, and 
a given {\it inverse temperature function} $\beta$.
\vskip0.1cm
\noindent
{\bf (ii)} We characterise the structure of these states and we find that, 
for the affine function $\tilde\b(h)=\beta(h)h=\frac{h-\m}{T}$ (with the Boltzmann constant $k=1$), hence including the chemical potential $\m$, they are reduced to the usual quasi-free 
equilibrium states.
\vskip0.1cm
\noindent
{\bf (iii)} We extend the study the BEC phenomenon to the setting of local equilibrium 
quasi-free states that, because of (ii), includes the equilibrium ones as particular cases.
\vskip0.1cm
\noindent
{\bf (iv)} In the setting of local equilibrium, we exhibit states exhibiting BEC, not only on the ground 
state, but also (or only) on some excited levels, depending on the function 
$\beta$, entering in the definition of the Local Equilibrium Principle.
\vskip0.1cm
\noindent
{\bf (v)}  We construct examples of local NESS for which BEC occurs 
also for spatial dimensions $d$ different from the usual ones ($d\geq3$ 
for free massive Bosons, and $d\geq2$ for massless particles).
\vskip0.1cm
\noindent
{\bf (vi)}  We construct local quasi-free NESS, exhibiting BEC in excited levels, 
for which the rotation symmetry can be spontaneously broken. The rotationally invariant ones are obtained by averaging on the spheres in momentum space $\{{\bf p}\in\br^d\mid\b(h({\bf p}))={\bf 0}\}$, whenever $h({\bf p})$ is the one-particle Hamiltonian of a free massive Boson.
\vskip0.1cm
\noindent
{\bf (vii)}  We extend our analysis to include the case of $q$-Deformed 
Commutation Relations with $q\in[-1,1]$, and we prove that BEC can occur 
also for the $q$-relations provided that $q\in(0,1]$. 
The case $q\in(0,1)$ opens some mathematically (and maybe also physically)
interesting possibilities which at the moment seem to be unexplored.
\vskip0.1cm
\noindent
{\bf (viii)}  Even in the usual case of equilibrium thermodynamics, we exhibit completely new states describing BEC, which are mathematically meaningful, and could have promising physical applications.
\vskip0.1cm
\noindent

Many of the results of the present paper concerning the $2^\text{nd}$ quantisation of free particles, are obtained by using a new approach based on the theory of the Distributions, which hopefully might be extended to more general situations.

\section{the local equilibrium condition}
\label{sucf}

As a first step in posing the problem, we consider a physical system with a finite degrees of freedom, whose observables are typically described by the $C^*$-algebra of all the bounded operators $\cb(\ch)$ acting on a separable Hilbert space $\ch$. The time-evolution is given in Heisenberg picture by
$$
a\in\cb(\ch)\mapsto \a_t(a)=e^{\imath Ht}ae^{-\imath Ht}\in\cb(\ch)\,,\quad t\in\br\,,
$$
where the Hamiltonian $H\geq0$ of the system is a self-adjoint operator acting on $\ch$, which is generally unbounded. To simplify the matter, we suppose that 
\begin{equation}
\label{hami}
H=\sum_{n=0}^{\infty}\eps_n|\psi_n\rangle\langle\psi_n|
\end{equation}
is positive with compact resolvent, with eigenvalues $\eps_n\in\s(H)$ and the corresponding eigenvectors $\psi_n\in\ch$, repeated according the multiplicity and rearranged in increasing order if needed. Fix any positive function
(i.e. the {\it local inverse temperature}) $\b:\br_+\rightarrow \br_+$ such that $e^{-\b(H)H}$ is a trace-class operator. According to the {\it Local Equilibrium Principle} (see e.g. \cite{AFQ} for details and the literature on the topic),
define the state
\begin{equation}
\label{lockms}
\om_\b(a):=\tr\big(e^{-\b(H)H}a\big)\,,\quad\a\in\cb(\ch)\,.
\end{equation}
Denote $\cf(\ch)$ the sub algebra consisting of all the finite-rank operators acting on $\ch$, and $v_{ij}:=|\psi_i\rangle\langle\psi_j|$, the $\psi_i$ being the eigenfunctions of $H$ according to \eqref{hami}. It is immediate to show that
$$
z\in\bc\mapsto\om_\b(a\a_z(b))\,,\quad a\in\cb(\ch)\,,\,\,b\in\cf(\ch)\,,
$$
is well defined and entire. The Local Equilibrium Principle (LEQ for short) simply means
\begin{equation*}
\om_\b(a\a_{t+i\b(\eps_i)}(v_{ij}))
=e^{[\b(\eps_i)-\b(\eps_j)]\eps_j}\om_\b(\a_{t}(v_{ij})a)\,,\quad a\in\cb(\ch)\,.
\end{equation*}
The LEP, which reduces to the usual KMS boundary condition (cf. \cite{BR, HHW})
when $\b$ is the constant function, is not immediately generalizable to arbitrary dynamical systems $(\ga,\a_t)$ consisting of a 
$C^*$-algebra, and an action $\a_t$ of one-parameter group of possibly outer $*$-automorphisms of $\ga$. A natural way to get such a possible generalisation at least for inner time-evolution, is to look at the modified evolution
\begin{equation}
\label{lockms2}
a\in\cb(\ch)\mapsto\a^{(\b)}_t(a):=e^{\imath\b(H)Ht}ae^{-\imath\b(H)Ht}\in\cb(\ch)\,.
\end{equation}
In the discrete spectrum Hamiltonian case discussed below, 
$$
z\in\bc\mapsto\om_\b(a\a^{(\b)}_z(b))\,,\quad a\in\cb(\ch)\,,b\in\cf(\ch)\,,
$$
is again well defined and entire, 
and the state $\om_\b$ satisfies the KMS boundary condition w.r.t. this modified evolution at inverse temperature $1/T=1$.

We end by reporting the property analogous to the well-known one concerning the KMS boundary condition in momentum space.
Consider any continuous compactly supported function $f$ on $\br$, together with its Fourier (inverse) transform $\check f$. Then for each $a\in\cb(\ch)$
$$
\m_a(f):=\int_{-\infty}^{+\infty}\check f(t)\om_\b(a^*\a^{(\b)}_t(a))\,,\quad
\n_a(f):=\int_{-\infty}^{+\infty}\check f(t)\om_\b(\a^{(\b)}_t(a)a^*)
$$
define Radon measures on $\br$. In this simple situation, 
$\m$ and $\n$ are linear combinations of Dirac measures supported in a subset 
$$
K_a\subset\{\b(\eps_n)\eps_n-\b(\eps_m)\eps_m\mid \eps_n,\eps_m\in\s(H)\}\,,
$$ 
depending of $a$. In addition, such measures 
are equivalent with 
Radon-Nikodym derivative given by
$$
\frac{\di \m_a}{\di \n_a}(k)=e^{-k}
$$
where the $k\in K_a$. The last simple calculation has the following meaning. The set the {\it Bohr Frequencies} of a physical system as before consists of all the possibly energies of the quanta that it can emit or absorb while interacts with the environment, see \cite{AK} for a simple explanation. It is given by $\{\eps_n-\eps_m\mid \eps_n,\eps_m\in\s(H)\}$. In the more general case of LEQ, the natural Bohr Frequencies (or the so-called {\it Arveson Spectrum} in the mathematical language) for the possible transition of the system are precisely those 
$\{\b(\eps_n)\eps_n-\b(\eps_m)\eps_m\mid \eps_n,\eps_m\in\s(H)\}$ of the modified Hamiltonian $H_\b:=\b(H)H$. We end the present section with the following facts. First, we observe the similar dependence of the temperature directly on the Arveson Spectrum of the involved automorphism arising from the spectrally passive states in \cite{dC}. Second, we point out again quasi-free states with non constant temperature described in \cite{BOR} for non equilibrium relativistic thermodynamics.

\section{local equilibrium principle for free Bosons}

Let $\gph$ the {\it one-particle Hilbert space}. The Canonical Commutation Relations  (CCR for short) algebra describes Bose particles and is generated by 
\begin{equation}
\label{cccrr}
a(f)a^{\dagger}(g)-a^{\dagger}(g)a(f)=\langle f|g\rangle\idd\,\quad f,g\in\gph_0\,.
\end{equation}
for a dense subspace $\gph_0\subset \gph$. It is well-known that the CCR cannot be represented as bounded operators acting on some Hilbert space even if $\gph=\bc$.
The standard procedure is to pass to the CCR represented in the Weyl form. It is the universal $C^*$-algebra $\ccr(\gph_0)$ generated by unitary operators satisfying the relations
\begin{align}
\label{cccrr1}
&W(h)^*=W(-h)\,,\quad W(0)=\idd\,,\\
&W(f)W(g)=e^{i\frac{\im(f,g)}{2}}W(f+g)\,,\quad f,g\in\gph_0\nn\,.
\end{align}
Thus, with $\ccr(\gph_0)$ we denote the universal $C^*$-algebra generated by the Weyl operators $\{W(f)\mid f\in\gph_0\}$ satisfying the relations \eqref{cccrr1}.

Let $h>0$ be the {\it one-particle Hamiltonian}. Suppose that $e^{\imath ht}\gph_0=\gph_0$. 
Then the time evolution is generated by a one-parameter group of Bogoliubov automorphisms which acts on the Weyl generators 
as
\begin{equation}
\label{tievo}
\a_t(W(f)):=e^{\imath t\di\G(h)}W(u)e^{-\imath t\di\G(h)}
=W(e^{\imath th}f)\,,\quad f\in\gph_0\,, t\in\br\,.
\end{equation}
Here, $H=\di\G(h)$ is the second quantised Hamiltonian, and we have identified $W(f)$ with its image $\pi_F(W(f))$ in the Fock representation because $\ccr(\gph_0)$ is a simple $C^*$-algebra, see e.g. \cite{BR, LRT} and the references cited therein.

For systems for which $\gph$ is finite-dimensional, the second quantised Hamiltonian $H:=\di\G(h)$
has compact resolvent, which typically occurs when one considers a free gas confined in a box of finite volume. In this framework, $e^{-\b(H)H}$ is automatically trace-class, see e.g. Proposition 5.2.27 of  \cite{BR}. Thus, the unique state $\om_\b$ satisfying the LEP w.r.t. the time evolution
\eqref{tievo} is the quasi-free state uniquely determined by the two-point function
$$
\om_\b(a^{\dagger}(g)a(f))=\big\langle f\big|(e^{\b(h)h}-1\big)^{-1}\big|g\big\rangle\,,\quad f,g\in\gph\,.
$$
By using the CCR \eqref{cccrr}, we have
\begin{align}
\label{lepccr}
&\om_\b(a(g)a^{\dagger}(e^{-\b(h)h}f))=\om_\b(a^{\dagger}(f)a(g))\,,\\
&\om_\b(a^{\dagger}(g)a(e^{\b(h)h}f))=\om_\b(a(f)a^{\dagger}(g))\,.\nn
\end{align}
The boundary conditions \eqref{lepccr} gives the natural generalisation for LEP in the cases of quasi-free states of CCR algebras, providing several nontrivial examples for which the dynamics is neither inner, nor has pure-point spectrum in general.

We start with a separable Hilbert space $\gph$, together with a self-adjoint positive Hamiltonian (in general unbounded) $h>0$, which are nothing but the one-particle Hilbert space with the relative one-particle Hamiltonian. Let $\eps\in\br\mapsto e(\eps)\in\cb(\ch)$ be the right-continuous resolution of the identity associated to $h$ (cf. Section VIII.3 of \cite{RS}) according to
\begin{equation}
\label{residy}
h=\int \eps\di e_h(\eps)\,.
\end{equation}
Fix a Borel function $\b:\br_+\to\br_+$ which is positive almost everywhere w.r.t. the measure class determined by the projection-valued measure $\di e(\eps)$. The key-point to manage open thermodynamical systems where exchange of matter is also allowed, is to introduce the {\it chemical potential}, see e.g. \cite{LL}. 
In order to achieve the chemical potential for local equilibrium, we follow the fact pointed out in Section \ref{sucf} which asserts that a state satisfying the LEP is indeed an equilibrium one for a modified Hamiltonian at inverse temperature $1/T=\b=1$. For the {\it activity} $z$ (cf. pag. 47 of \cite{BR}), it leads to $z(1/T,\m)=z(1,\m)=e^{\m}$, where $\m$ is the chemical potential. This justifies the definition for $\m\in\br$ in
\begin{equation}
\label{chpoz}
\g_\m(x):=e^{\b(x)x-\m}\,.
\end{equation}
For the $q$-particles with $q\in(0,1]$ we have a natural restriction to the value of the allowed chemical potentials $\m$ coming from the positivity condition of the occupation number, see below. For the Bose case $q=1$, such a condition simply leads to $\m\leq0$. 

A {\it quasi-free} state is a mean-zero gaussian state $\om$, a-priori defined on the whole 
$\ccr(\gph)$. Is uniquely determined by its two-point function as it is explained in Section 5.2.3 of \cite{BR}. In fact, we can start with a quadratic form $Q$ on $\gph$ such that
$$
\om_Q(a^\dagger(f)a(f))=Q(f)\,,\quad f\in\gph\,.
$$
This means nothing but that in a more rigorous form,
$$
\om_Q(W(f))=\exp\{-(\|f\|^2/4+Q(f)/2)\}\,,\quad f\in\gph\,,
$$
which is meaningful also when $Q(f)=+\infty$. A standard procedure to manage quasi-free states (see Theorem 5.2.31 of \cite{BR}, for the situation describing BEC) is to consider the domain $D_Q\subset\gph$ made of all the elements on which $Q$ is finite, and consider the sesquilinear form $F_Q$ defined by polarisation on all of $D_Q$ such that
$$
\om_Q(a^\dagger(f)a(g))=F(g,f)\,,\quad f,g\in D_Q\,.
$$
Thus,
we fix a value of $\m\leq0$, and a dense subspace $\gph_0\subset\gph$ such that
\begin{itemize}
\item[(i)] $e^{\imath ht}\gph_0=\gph_0$, $t\in\br$.
\end{itemize}
Among the quasi-free states $\om_Q$ on $\ccr(\gph_0)$, we consider those for which 
\begin{itemize}
\item[(ii)] $\gph_0\subset D_Q$, and $\g_\m(h)\gph_0\subset D_Q$, with $\g_\m$ in \eqref{chpoz}.
\end{itemize}
The natural generalisation of the LEP to quasi-free states of Bosonic systems is one of the (equivalent) conditions given in \eqref{lepccr}.
\begin{Dfn}
\label{lepcafi}
Suppose that the dense subspace $\gph_0\subset\gph$ and the quasi-free state $\om\in\cs(\ccr(\gph_0))$ fulfill (i)-(ii) above. Then $\om$ satisfies the Local Equilibrium Principle at local inverse temperature $\b$ and chemical potential $\m\leq0$ if 
\begin{equation}
\label{lepccr2}
\om(a^{\dagger}(f)a(\g_\m(h)g))=\om(a(g)a^{\dagger}(f))\,,\quad f,g\in\gph_0\,.
\end{equation}
\end{Dfn}
It is immediate to prove (cf. Section 5.3.1 of \cite{BR}) that a quasi-free state satisfying the LEP is invariant w.r.t. the modified evolution \eqref{lockms2}, provided 
$e^{\imath \b(h)ht}\gph_0=\gph_0$, $t\in\br$. In general, it is not clear if it is invariant also w.r.t. the natural time-evolution of the system generated by the one-parameter Bogoliubov automorphisms $t\mapsto e^{\imath ht}$. However, it is possible to see by direct inspection for all the cases under consideration in the present paper, that those are indeed invariant w.r.t. the dynamics generated by the one-particle Hamiltonian.
Another nontrivial question arises in introducing the chemical potential $\m$ in order to investigate the thermodynamics of open systems. Our ansatz is to take for the occupation number $n_\eps$ at energy $\eps$ and chemical potential $\m$ (independent on the energy levels),
$$
n_\eps=\frac1{e^{\b(\eps)\eps-\m}-1}\,.
$$
Another reasonable choice would be
$$
n_\eps=\frac1{e^{\b(\eps-\m)(\eps-\m)}-1}\,.
$$
Both choices leads to equivalent results in the usual equilibrium situation (i.e. if $\b=\text{const.}$), but these might lead to different situations in non equilibrium ones. The most general situation in introducing the chemical potential would be that it also depends on the energy levels. Another quite natural choice will be briefly discussed in Section \ref{outlook}. 

As we are going to see, for the Bose case and also for $q$-particles with $q>0$, we can exhibit a multitude of quasi-free states satisfying LEP and exhibiting BEC. Among such states, we will find many ones which exhibit unexpected properties, even in the equilibrium situation. Due to Pauli exclusion principle, the Fermi case and also $q$-particles with $q\leq0$, leads to only one stationary quasi-free state satisfying LEP. In fact, the same analysis can be done {\it mutatis-mutandis} for Fermi particle for which the condensation never occur. Consider the Canonical Anti-commutation Relation algebra which is the $C^*$-algebra $\carf(\gph)$ generated by annihilators $c(f)$ satisfying
 \begin{equation}
 \label{ccarr}
 c(f)c^{\dagger}(g)+c^{\dagger}(g)c(f)=\langle f|g\rangle\idd\,\quad f,g\in\gph\,.
 \end{equation}
For the gas of free Fermions, fix a density $\r>0$ and consider the (unique) solution of the equation
$$
\int_{\br_+}\frac{\di N(\eps)}{e^{\b(\eps)\eps-\m}+1}=\r
$$
in the unknown $\m\in\br$. Here, the cumulative function $N$ is the so-called {\it integrated density of the states} and describes the density of eigenvalues of the Hamiltonian in the infinite volume limit. At the knowledges of the authors, the integrated density of the states might not exists for a chosen finite volume exhaustion of the underlying physical space, or it may depend on the chosen exhaustion even in the amenable cases. The reader is referred to \cite{F1} for the rigorous definition of the integrated density of the states and some relevant properties for some relevant amenable and non amenable cases.
Then the quasi-free state $\om_{\b,\m}$ determined by the two-point function
$$
\om_{\b,\m}(c^\dagger(f)c(g))=\big\langle g\big|(e^{\b(h)h-\m}+1)^{-1}\big|f\big\rangle\,,\quad f,g\in\gph
$$
is the unique quasi-free state fulfilling the LEP for the inverse local temperature function $\b$ and the chemical potential 
$\m$.

We end the present section by pointing out that our framework of local equilibrium can be easily generalised to self-adjoint, not necessarily semi-bounded Hamiltonians $h$, provided that 
for the local temperature function $\b$, $\b(x)x$ is semi-bounded, almost surely w.r.t. the measure class determined by the resolution of the identity $\eps\mapsto e_h(\eps)$ of $h$ in \eqref{residy}.

\section{local equilibrium principle and the Bose-Einstein condensation}
\label{00000}

In order to exhibit states describing the BEC even in the more general context of local equilibrium, we specialise the matter to the simplest model describing non relativistic free Bosons living on $\br^d$. To simplify, we put $m=1/2$ for their mass. Analogous considerations can be done for Bosons on lattices $\bz^d$. Indeed, fix the functions in the class 
$$
\check{\cd}(\br^d)\subset\cs(\br^d)\subset L^2(\br^d,\di^d{\bf x})
$$ 
made of the Fourier anti-transform of all the infinitely often differentiable functions with compact support in momentum space.
The corresponding one-particle Hamiltonian will be 
$$
h=-\sum_{j=1}^d\frac{\partial^2\,\,\,\,}{\partial x^2_j}\equiv-\D
$$
given by the opposite of the Laplace operator on $\br^d$. It is immediate to see that (i) before Definition \ref{lepcafi} is satisfied. The one-particle Hamiltonian $-\D$ is nothing but the multiplication for the function
$$
k^2:=\sum_{j=1}^dk_j^2
$$
in the momentum space after Fourier Transform, where as usual, ${\bf k}=(k_1,\dots,k_d)$. Thus, we can consider any non negative Borel function $\b:(0,+\infty)\rightarrow (0,+\infty)$ with $\b>0$ almost everywhere w.r.t. the Lebesgue measure on $\br$. 

The first interesting phenomenon is that, in the setting of local equilibrium, particles can condensate also on excited levels of the energies. The main result of the present section concerning the condensation regime for which $\m=0$ (cf. Section \ref{new}) is summarised in the following
\begin{Thm}
\label{tcbb}
Let $\frac1{e^{\b(p^2)p^2}-1}\in L^1_{\rm loc}(\br^d)$, and the function $x\b(x)\in L^\infty_{\rm loc}(\br_+)$. Suppose that for some $x_0\in[0,+\infty)$, $\lim_{x\to x_0}x\b(x)=0$. For each $D\geq0$, consider the point-mass measure $\n_{D,{\bf k}}:=D\d_{\bf k}$. Then
the quasi-free state $\om_{\b,\n_{D,{\bf k}}}\in\cs(\ccr(\check{\cd}(\br^d)))$ with two-point function
\begin{equation}
\label{1ab1}
\om_{\b,\n_{D,{\bf k}}}(a^\dagger(\check f)a(\check g)):=\int_{\br^d}\frac{f({\bf p})\overline{g({\bf p})}}{e^{\b(p^2)p^2}-1}\di^d{\bf p}
+D f({\bf k})\overline{g({\bf k})}\,,\quad f,g\in\cd(\br^d)\,,
\end{equation}
satisfies the LEP w.r.t. the local inverse temperature function $\b$ and chemical potential $\m=0$, provided $k^2=x_0$.
\end{Thm}
\begin{proof}
Thanks to $\frac1{e^{\b(p^2)p^2}-1}\in L^1_{\rm loc}(\br^d)$, \eqref{1ab1} is well defined for each $f,g\in\cd(\br^d)$. In addition, as $x\b(x)\in L^\infty_{\rm loc}(\br_+)$ then 
$$
\int_{\br^d}\frac{e^{2\b(p^2)p^2}|f({\bf p})|^2}{e^{\b(p^2)p^2}-1}\di^d{\bf p}\leq e^{2\|x\b(x)\lceil_{[0,r^2]}\|_\infty}\|f\|^2_\infty
\int_{p\leq r}\frac{\di^d{\bf p}}{e^{\b(p^2)p^2}-1}\di^d{\bf p}<+\infty\,,
$$
where $\supp(f)\subset D_r$, $D_r$ being the disk of radius $r$ centred in the origin. Finally, if $k^2=x_0$, the function $\widehat{e^{\b(h)h}\check f}$ is uniquely defined in ${\bf k}$ as
$$
\big(\widehat{e^{\b(h)h}\check f}\big)({\bf k})=\big(\lim_{{\bf p}\to{\bf k}}e^{\b(p^2)p^2}\big)f({\bf k})=
\big(\lim_{x\to x_0}e^{\b(x)x}\big)f({\bf k})=f({\bf k})\,,
$$ 
because the function $e^{\b(p^2)p^2}$ coincides a.e. with a measurable function which is continuous in $k^2=x_0$. Collecting together, we have first that $e^{\b(h)h}\check f$ is in the domain of the form \eqref{1ab1}. In addition, by using the commutation relation \eqref{cccrr}, we compute
\begin{align*}
&\om_{\b,\n_{D,{\bf k}}}(a(\check g)a^\dagger(\check f))
=\int_{\br^d}\bigg(1+\frac1{e^{\b(p^2)p^2}-1}\bigg)
f({\bf p})\overline{g({\bf p})}\di^d{\bf p}+Df({\bf k})\overline{g({\bf k})}\\
=&\int_{\br^d}\frac{e^{\b(p^2)p^2}}{e^{\b(p^2)p^2}-1}
f({\bf p})\overline{g({\bf p})}\di^d{\bf p}+De^{\b(k^2)k^2}f({\bf k})\overline{g({\bf k})}
=\om_{\b,\n_{D,{\bf k}}}\big(a^\dagger(\check f)a\big(e^{\b(h)h}\check g\big)\big)\,,
\end{align*}
that is \eqref{lepccr2} is satisfied for $\m=0$.
\end{proof}
The two boundedness conditions in Theorem \ref{tcbb} have a different meaning. The second one $x\b(x)\in L^\infty_{\rm loc}(\br_+)$, automatically verified in the usual equilibrium setting when $\b=\text{const.}$, can be relaxed case-by-case when one considers special examples of temperature function $\b(\eps)$. The first one 
$\frac1{e^{\b(p^2)p^2}-1}\in L^1_{\rm loc}(\br^d)$ has an important physical meaning. First of all, we note that $\frac1{e^{\b p^2}-1}\in L^1_{\rm loc}(\br^d)$ is equivalent to
$\frac1{e^{\b p^2}-1}\in L^1(\br^d)$ in the equilibrium situation. This is nothing but $\r_c(\b)<+\infty$ where $\r_c(\b)$ is the {\it critical density} for the function $\b$, which does not depend on the energy levels in equilibrium thermodynamics. A key-role to study the appearance of the BEC is to look at the {\it critical density} $\r_c(\b)$ at inverse temperature $\b=1/T$, see \cite{BR}. In our situation (including the equilibrium case), the critical density $\r_c(\b)$ for the function $\b$ is given by
\begin{equation}
\label{rhocr}
\r_c(\b)=\int_{\br^d}\frac{\di^d{\bf p}}{e^{\b(p^2)p^2}-1}\,.
\end{equation}
We will also show in Section \ref{new} the natural role played by the critical density in studying the condensation effects. In the more general situation of LEP,
$\frac1{e^{\b(p^2)p^2}-1}\in L^1_{\rm loc}(\br^d)$ is in general weaker than $\r_c(\b)<+\infty$. This means that we can easily exhibit quasi-free states describing condensation effects, for which $\r_c(\b)=+\infty$. A similar phenomenon can happen in studying BEC in equilibrium thermodynamics for inhomogeneous systems. We show that this last phenomenon is of different nature than the analogous one described in \cite{F2, F3}. In order to do that, we start by looking at the local density of particles $\r_\om({\bf x})$
of a quasi-free state $\om$. It is given by
$$
\r_\om({\bf x})=\om(a^\dagger(\d_{\bf x})a(\d_{\bf x}))\,,
$$ 
where $\d_{\bf x}$ is the Dirac distribution centred in ${\bf x}\in\br^d$, provided that the r.h.s. is meaningful, otherwise it is infinite. For the state in \eqref{1ab1}, we get for the local density of particles
$$
\r_{\om_{\b,\n_{D,{\bf k}}}}({\bf x})
=\int_{\br^d}\frac{|\widehat{\d_{{\bf x}}}({\bf p})\big|^2}{e^{\b(p^2)p^2}-1}\di^d{\bf p}+D|\widehat{\d_{{\bf x}}}({\bf 0})\big|^2\,.
$$
From this simple calculation, we conclude that the $\r_{\om_{\b,D}}({\bf x})$ is homogeneous and assume the form
$$
\r_{\om_{\b,\n_{D,{\bf k}}}}({\bf x})=\r_c(\b)+\r_\text{cond}(\om_{\b,\n_{D,{\bf k}}})\,,
$$
where
$$
\r_\text{cond}(\om_{\b,\n_{D,{\bf k}}})=\lim_{\La\uparrow\br^d}\frac D{\vol(\La)}\int_{\La}\big|\widehat{\d_{{\bf x}}}({\bf 0})\big|^2\di^d{\bf x}=D\,.
$$
Thus, the local density of particles $\r_{\om_{\b,D}}({\bf x})$ of the state $\om_{\b,D}$ is finite if and only if the critical density \eqref{rhocr} of the model is finite too. This means nothing else that the models for which the critical density is infinite, still provide states exhibiting BEC effects, but all of them have infinite local density of particles.
Namely, such states are in some sense unphysical. 

We pass to investigate the role played by the natural action $\a_R$, $R\in O(d)$ of the rotation group on $\ccr(\check{\cd}(\br^d))$. For non trivial $R\in O(d)$ and non zero ${\bf k}$ as in Theorem \ref{tcbb}, we have $R{\bf k}\neq {\bf k}$ but $\|R{\bf k}\|^2=x_0$. Thus, if there exists $x_0>0$ such that $\lim_{x\to x_0}\b(x)x=0$, the rotation symmetry is spontaneously broken. With $\n$ any bounded positive Radon measure on the sphere $\bs_k\subset\br^d$ of radius $k$, it is easy to see that for $f,g\in\cd(\br^d)$, the quasi free-state with two-point function  
\begin{equation}
\label{rotcinv}
\om_{\b,\n}(a^\dagger(\check f)a(\check g)):=\int_{\br^d}\frac{f({\bf p})\overline{g({\bf p})}}{e^{\b(p^2)p^2}-1}\di^d{\bf p}
+\int_{\bs_k}f({\bf p})\overline{g({\bf p})}\di\n({\bf p}) 
\end{equation}
satisfies the LEP. We immediately see that
$$
\r_{\om_{\b,\n}}({\bf x})
=\r_c(\b)+\n(\br^d)\,,
$$
where the last addendum on the r.h.s. takes into account of the portion of the condensate. In general, $\om_{\b,\n}$ is not rotationally invariant. However, we get
\begin{Prop}
\label{pprroo}
With $D\geq0$ and $\Om_k$ the normalised rotationally invariant measure on the sphere $\bs_k$, the states $\om_{\b,D\Om_k}$ in \eqref{rotcinv}
are rotationally invariant.
\end{Prop}
\begin{proof}
Let $R\in O(d)$ be an orthogonal matrix with $R^t$ its transpose one. Define $f_R(x):=f(R^tx)$. We easily get
$\widehat{f_R}=(\hat f)_{R}$. By taking into account that the function $\frac1{e^{\b(p^2)p^2}-1}$ is rotationally invariant, we get after an elementary change of variable,
\begin{align*}
&\om_{\b,D\Om_k}\big(\a_R(a^\dagger(\check f)a(\check g))\big)=
\om_{\b,D\Om_k}(a^\dagger((\check f)_R)a((\check g)_R))\\
=&\int_{\br^d}\frac{f(R^t{\bf p})\overline{g(R^t{\bf p})}}{e^{\b(p^2)p^2}-1}\di^d{\bf p}
+D\int_{\bs_k}f(R^t{\bf p})\overline{g(R^t{\bf p})}\di\Om_k({\bf p})\\
=&\int_{\br^d}\frac{\hat f({\bf p})\overline{\hat g({\bf p})}}{e^{\b(p^2)p^2}-1}\di^d{\bf p}
+D\int_{\bs_k}f({\bf p})\overline{g({\bf p})}\di\Om_k({\bf p})\\
=&\om_{\b,D\Om_k}(a^\dagger(\check f)a(\check g))\,.
\end{align*}
\end{proof}

We already have discussed that the existence of quasi-free states exhibiting BEC for which the local density is finite, is determined by the convergence of the integral in \eqref{rhocr} describing the critical density. To see that we can have the condensation even for spatial dimensions different from the usual one $d\geq3$ (or $d\geq2$ for the fotonic/phononic Hamiltonian 
$h({\bf k})=k$), we consider simple examples for which $\b\in C\big((0,+\infty)\big)$, and
$$
\b(x)\approx x^{\a_0}\,\,\text{for}\,\,x\to 0^+\,;\quad \b(x)\approx \a_\infty\frac{\ln x}x\,\,\text{for}\,\,x\to+\infty\,.
$$
The condition $\r_c(\b)<+\infty$ leads to $\a_\infty>0$, whereas $\lim_{x\downarrow 0}\b(x)x=0$ leads to $\a_0+1>0$. Concerning the critical density, we compute for the one-particle Hamiltonian of the form $h({\bf k})=k^s$, $s\geq1$ (to avoid unphysical models, even if the last restriction plays no technical role),
$$
\int^{+\infty}_{0}\frac{\di^d{\bf p}}{e^{\b(p^s)p^s}-1}
\approx\int^{1}_{0}p^{d-1-s(\a_0+1)}\di p
+\int^{+\infty}_{1}p^{d-1-s\a_\infty}\di p\,,
$$
which converges if and only if 
$$
s(\a_0+1)<d<s\a_\infty\,.
$$
This means that for the NESS considered in the present paper, it is possible to have or not  BEC of free Bosons on $\br^d$ for dimensions different from those one finds in equilibrium thermodynamics.

\section{a new approach to the condensation}
\label{new}

The standard way to find states exhibiting BEC on a locally compact manifold $M$, which can be $\br^d$ (see e.g. \cite{BR}) 
or an infinitely extended network (cf. \cite{F2, F3}), is to start from the Bose-Gibbs grand canonical ensemble (cf. \cite{LL}) of the finite volume theories based on a fixed exhaustion $\{\La_n\}_{n\in\bn}$, $\La_n\uparrow M$, together with the associated sequence of finite volume Hamiltonians $\{H_{\La_n}\}_{n\in\bn}$. One considers the finite volume density
$$
\r_{\La}(\b,\m):=\int_{\br_+}\frac{\di N_{\La}(\eps)}{e^{\b(\eps)\eps-\m}-1}
$$
associated to a compact region $\La\subset M$. Here, with $H_\La=\int\eps\di e_\La(\eps)$, and
$N_{\La}$ is the cumulative function obtained as
$$
N_\La(\eps)=\frac{\tr(E_\La((-\infty,\eps]))}{\vol(\La)}\,,
$$ 
which is well defined provided $H_\La$ has compact resolvent. Under the usual conditions imposed to the function $\b$, the standard properties of the Hamiltonian $H$ and the associated finite volume ones $H_\La$, the function $\r_{\La}(\b,\m)$ has the behaviour described in Fig. \ref{fig}, with $\eps_0(\La)>0$ the finite volume ground state energy of the system.
\begin{figure}
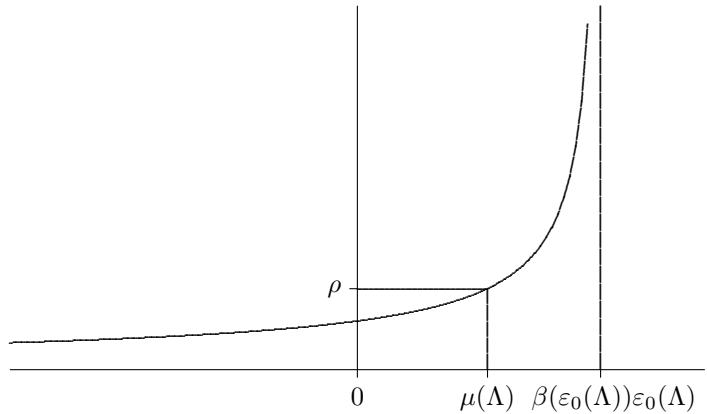

\hbox to\hsize\bgroup\hss
\beginpicture
\setcoordinatesystem units <1.8truein,.2truein>
\setplotarea x from -1 to 1, y from 0 to 9.5
\axis bottom ticks short withvalues {$0$} {$\m(\La)$} {$\quad\b(\eps_0(\La))\eps_0(\La)$}
/ at 0 0.3745455  0.7 / /
\axis left shiftedto x=0 ticks short withvalues {$\rho$} / at 2.1135204 / /
\plot  0 2.1135204  0.3745455 2.1135204 /
\plot 0.7 0  0.7 9.5 /
\plot  0.3745455 0  0.3745455 2.1135204 /
\plot
-1.           0.7020489
  -0.9830303    0.7067601
  -0.9660606    0.7115512
  -0.9490909    0.7164242
  -0.9321212    0.7213816
  -0.9151515    0.7264256
  -0.8981818    0.7315588
  -0.8812121    0.7367835
  -0.8642424    0.7421026
  -0.8472727    0.7475187
  -0.8303030    0.7530348
  -0.8133333    0.7586537
  -0.7963636    0.7643787
  -0.7793939    0.7702129
  -0.7624242    0.7761597
  -0.7454545    0.7822226
  -0.7284848    0.7884054
  -0.7115152    0.7947118
  -0.6945455    0.8011458
  -0.6775758    0.8077116
  -0.6606061    0.8144137
  -0.6436364    0.8212565
  -0.6266667    0.8282449
  -0.6096970    0.8353839
  -0.5927273    0.8426788
  -0.5757576    0.850135
  -0.5587879    0.8577584
  -0.5418182    0.8655550
  -0.5248485    0.8735312
  -0.5078788    0.8816938
  -0.4909091    0.8900497
  -0.4739394    0.8986065
  -0.4569697    0.9073719
  -0.44         0.9163543
  -0.4230303    0.9255624
  -0.4060606    0.9350053
  -0.3890909    0.9446927
  -0.3721212    0.9546351
  -0.3551515    0.9648433
  -0.3381818    0.9753288
  -0.3212121    0.9861039
  -0.3042424    0.9971817
  -0.2872727    1.0085759
  -0.2703030    1.0203012
  -0.2533333    1.0323734
  -0.2363636    1.0448092
  -0.2193939    1.0576263
  -0.2024242    1.0708439
  -0.1854545    1.0844823
  -0.1684848    1.0985635
  -0.1515152    1.1131108
  -0.1345455    1.1281495
  -0.1175758    1.1437067
  -0.1006061    1.1598115
  -0.0836364    1.1764955
  -0.0666667    1.1937928
  -0.0496970    1.2117401
  -0.0327273    1.2303776
  -0.0157576    1.2497487
    0.0012121    1.2699007
    0.0181818    1.2908852
    0.0351515    1.3127586
    0.0521212    1.3355828
    0.0690909    1.3594257
    0.0860606    1.384362
    0.1030303    1.4104742
    0.12         1.4378536
    0.1369697    1.4666016
    0.1539394    1.496831
    0.1709091    1.5286679
    0.1878788    1.5622534
    0.2048485    1.5977467
    0.2218182    1.6353271
    0.2387879    1.6751982
    0.2557576    1.7175919
    0.2727273    1.7627737
    0.2896970    1.8110489
    0.3066667    1.8627711
    0.3236364    1.9183517
    0.3406061    1.9782727
    0.3575758    2.0431033
    0.3745455    2.1135204
    0.3915152    2.1903362
    0.4084848    2.2745348
    0.4254545    2.3673211
    0.4424242    2.470188
    0.4593939    2.5850104
    0.4763636    2.7141784
    0.4933333    2.8607926
    0.5103030    3.0289567
    0.5272727    3.2242303
    0.5442424    3.454357
    0.5612121    3.7304879
    0.5781818    4.069352
    0.5951515    4.4973694
    0.6121212    5.059147
    0.6290909    5.8371213
    0.6460606    7.0048035
    0.6630303    9.0109256
    /
\endpicture
\hss\egroup
 \caption{{The finite volume chemical potential at fixed density $\r$.}}
     \label{fig}
\end{figure}
This means that, for each fixed density $\r>0$, the equation 
\begin{equation*}
\int_{\br_+}\frac{\di N_{\La}(\eps)}{e^{\b(\eps)\eps-\m}-1}=\r
\end{equation*}
in the unknown $\m$ has a unique solution 
$\m(\La)\in(-\infty,\b(\eps_0(\La))\eps_0(\La))$, which allows to determine the sequences $\{\m(\La_n)\}$ of the finite volume chemical potentials. The critical density is now defined as
$$
\r_c(\b):=\int_{\br_+}\frac{\di N(\eps)}{e^{\b(\eps)\eps}-1}\,,
$$
where $N$ is the Integrated Density of the States associated to the infinite volume Hamiltonian of the system (see \cite{F1} for the rigorous definition), which exists for most of the models of interests with natural choices of exhaustions. If $\r_c<+\infty$ and the fixed density in Fig. \ref{fig} of the system $\r>\r_c$, $\m(\La_n)\to0$ (possibly by passing to a converging subsequence if necessary), and the condensation phenomena take place. The standard approach briefly outlined here and commonly used to construct states exhibiting condensation, even if has a very clear physical justification, is extremely sensitive of the boundary conditions imposed to the finite volume Hamiltonians (see e.g. \cite{BR} and the literature cited therein), and a-priori also of the chosen finite volume exhaustion $\{\La_n\}_{n\in\bn}$ (see e.g. \cite{VLP} and the literature cited therein). Due to the new phenomena relative to the BEC arising from the LEP described in Section \ref{00000}, it is expected that this approach might be not flexible enough to select/construct explicitly all the possible states exhibiting BEC. For example, if one fix the density as described above, one might select only rotationally invariant states which exhaust all of them in the equilibrium situation but not in the more general setting of local equilibrium (cf. Theorem \ref{tcbb}). In addition, even in equilibrium thermodynamics, we will find states exhibiting BEC but for which the density of particles is infinite (cf. Proposition \ref{ddeldercs}). This suggests new approaches to the investigation of the BEC, which is the goal of the present section. We will see that this approach provides new and unexpected examples of states exhibiting condensation (cf. Section \ref{0nw1}), even in the usual situation of equilibrium thermodynamics where the temperature is fixed. 

For our purposes, following the model described in Section \ref{00000} we consider $h({\bf k})=k^s$. In order to exclude unphysical models, we suppose $s\geq1$, including the massive non relativistic Hamiltonian $h({\bf k})=k^2$ (with the normalised particle mass $m=1/2$), and the massless relativistic one $h({\bf k})=k$ (corresponding to the normalised speed of the light $c=1$, or the velocity of the sound 1 in the case of phonons).
In order to avoid technicalities, we make some reasonable restrictions to the function $\b:\br_+\to\br_+$. Put $\tilde\b(x):=\b(x)x$ and assume 
\begin{itemize}
\item[(i)] $\tilde\b\in C([0,+\infty))$, with $\inf_{\br_+}\tilde\b=0$.
\item[(ii)] If $E:=\tilde\b^{-1}(\{0\})$, then $E\subset[0,+\infty)$ is made of at most a finite numbers of points, with the empty set allowed. 
\item[(iii)] We assume also that
$$
\int^{+\infty}_{a}\frac{x^{\frac{d}s-1}}{e^{\tilde\b(x)}-1}\di x<+\infty\,,
$$
for each $a>\max E$.
\end{itemize}
Assuming the continuity of temperature function is just to avoid technicalities, yet providing nontrivial examples, whereas $\inf_{\br_+}\tilde\b=0$ can be assumed without loosing generality just by passing to the new function $\tilde\b-\inf\tilde\b$. The second condition is to avoid too general situations, perhaps already considered in  Section \ref{00000}. The third one says that $\frac1{e^{\b(k^s)k^s-1}}\in L^1_\text{loc}(\br^d)$ is simply equivalent to $\frac1{e^{\b(k^s)k^s-1}}\in L^1(\br^d)$ as for equilibrium thermodynamics.
Denote
$$
\bs_r:=\{({\bf p},{\bf k})\in\br^d\times\br^d\mid {\bf p}^2={\bf k}^2=r^2\}\,.
$$
the sphere of radius $r$ living in the diagonal of the product space $\br^d\times\br^d$. As our approach is quite general, we can manage the general $q$-Commutation Relations for 
$q\in[-1,1]$, $\pm1$ being the Bose-Fermi alternative. By using Fourier Transform, we can consider the density creators and annihilators in momentum space which are operator-valued distributions, by putting 
\begin{equation}
\label{cfzx}
a^\dagger(\check f)=\int_{\br^d}f({\bf k})a^\dagger({\bf k})\di^d{\bf k}\,,\quad a(\check g)=\int_{\br^d}\overline{g({\bf k})}a({\bf k})\di^d{\bf k}\,.
\end{equation}
Then the $q$-Commutation Relations, corresponding to \eqref{cccrr}, \eqref{ccarr} for the Bose-Fermi cases, can be rewritten as
\begin{equation}
\label{ccqrr}
a({\bf k})a^\dagger({\bf p})-qa^\dagger({\bf p})a({\bf k})=\d({\bf k}-{\bf p})\idd\,.
\end{equation}
It can be proven that the commutation relations \eqref{ccqrr}, still generate an abstract $C^*$-algebra, even for $q\in(-1,1)$, see e.g. \cite{JSW}. In all the situations 
$q\in[-1,1]$ considered here including the Bose/Fermi situation $q=\pm1$, we denote by $\ccr_q(\gph)$ the involved $C^*$-algebra. In addition,
the quasi-free states are also well defined for the deformed cases $-1<q<1$ because their  $2n+1$-point functions is $0$, and $2n$-point functions are described by the so-called $q$-{\it determinant} (being the case $q=-1$ indeed a determinant, known as the Slater determinant)
which is seen to be positive for the case considered here, see \cite{BS, Li, OTW}.
The LEP for $q$-particles, $q\in[-1,1]$, still assumes the form \eqref{lepccr2}
in Definition \ref{lepcafi}.

The new approach in searching quasi-free states $\om$ exhibiting condensation is to look at those for which the two-point function is given by a distribution $F_\om$,
at least for the situation of free particles considered in the present paper. For such a purpose, we use the conventional integration symbology of the theory of Distributions (cf. \cite{Vl})
for the natural pairing $F(f)$ between smooth functions $f$ and elements $F$ of the topological duals
$$
f\times F\in\cd(\br^l)\times\cd'(\br^l)\mapsto F(f)=:\int_{\br^l}F(\boldsymbol{\xi})f(\boldsymbol{\xi})\di^l\boldsymbol{\xi}\,.
$$
With this symbology, for $f,g\in\cd(\br^d)$, we compute by considering \eqref{cfzx},
\begin{align*}
\om(a^\dagger(\check f)a(\check g))
=&\int_{\br^d\times\br^d}\om(a^\dagger({\bf p})a({\bf k}))
f({\bf p})\overline{g({\bf k})}\di^d {\bf p}\di^d {\bf k}\\
=&\int_{\br^d\times\br^d}F_\om({\bf p},{\bf k})f({\bf p})\overline{g({\bf k})}\di^d {\bf p}\di^d {\bf k}\,,
\end{align*}
where $F_\om\in\cd'(\br^d\times\br^d)$ is some distribution given by $F_\om({\bf p},{\bf k}):=\om(a^\dagger({\bf p})a({\bf k}))$.  
By positivity, such a distribution should satisfy for each $f\in\cd(\br^d)$,
$$
\int_{\br^d\times \br^d}F_\om({\bf p},{\bf k})f({\bf p})\overline{f({\bf k})}\di^d {\bf p}\di^d {\bf k}\geq0\,.
$$
A positive definite distribution as above is said to be a {\it kernel}. In addition, a kernel is automatically real, that is it should satisfy in the sense of distribution,
\begin{equation}
\label{cerelz}
F_\om({\bf k},{\bf p})=\overline{F_\om({\bf p},{\bf k})}\,,
\end{equation}
where the bar stands for complex conjugation. 

Let $\om\in\cs(\ccr_q(\check\cd(\br^d))$ be a quasi-free state whose two-point function satisfies \eqref{lepccr2}.
Suppose further that its two-point function is given by a kernel $F_\om$ as described above. When we try to impose condition \eqref{lepccr2} to $F_\om$, a product of a distribution with a function which is in general not smooth, shall appear. This product does not define any distribution in general. Thus, we suppose further that 
$e^{(\b(h({\bf p}))h({\bf p})-\m)}F_\om({\bf p},{\bf k})$ defines still a distribution. Notice that, by positivity, $e^{(\b(h({\bf k}))h({\bf k})-\m)}F_\om({\bf p},{\bf k})$ also defines a distribution. This will play a crucial role in the following.
By using the commutation rule \ref{ccqrr}, the LEP 
\eqref{lepccr2} leads for the chemical potential $\m$, and the parameter $q\in[-1,1]$ to
\begin{align*}
&\int_{\br^d\times\br^d}\big(e^{\b(h({\bf k}))h({\bf k})-\m}\big)\om(a^\dagger({\bf p})a({\bf k}))
f({\bf p})\overline{g({\bf k})}\di^d {\bf p}\di^d {\bf k}\\
=&\om(a^\dagger(\check f)a(\g_\m(h)\check g))
=\om(a(\check g)a^\dagger(\check f))
=\langle g|f\rangle+q\om(a^\dagger(\check f)a(\check g))\\
=&\int_{\br^d\times\br^d}(\d({\bf p}-{\bf k})+q\om(a^\dagger({\bf p})a({\bf k})))
f({\bf p})\overline{g({\bf k})}\di^d {\bf p}\di^d {\bf k}\,,
\end{align*}
obtaining
$$
\int_{\br^d\times\br^d}\bigg[\big(e^{\b(h({\bf k}))h({\bf k})-\m}-q\big)F_\om({\bf p},{\bf k})-\d({\bf p}-{\bf k})\bigg]
f({\bf p})\overline{g({\bf k})}\di^d {\bf p}\di^d {\bf k}=0\,.
$$
As the last should be satisfied for the total set in $\cd(\br^d\times\br^d)$ generated by elementary tensors, we obtain
\begin{equation}
\label{cedcae}
\big(e^{\b(h({\bf k}))h({\bf k})-\m}-q\big)F_\om({\bf p},{\bf k})=\d({\bf p}-{\bf k})\,.
\end{equation}
Combining the reality condition \eqref{cerelz} with \eqref{cedcae}, we get also
\begin{equation}
\label{cedcae1}
\big(e^{\b(h({\bf p}))h({\bf p})-\m}-q\big)F_\om({\bf p},{\bf k})=\d({\bf p}-{\bf k})\,.
\end{equation}
The above computations tells us nothing but that the kernel $F_\om$ reproducing a quasi-free state on $\ccr_q(\check\cd(\br^d))$, which solves \eqref{cedcae} (or equivalently \eqref{cedcae1}) are the natural candidate to describe those satisfying the LEP for the $q$-deformed relations at the inverse temperature function $\b$ and chemical potential $\m$. 
As we will see below, in searching solutions of \eqref{cedcae}, 
we should divide for the function $e^{(\b(h({\bf p}))h({\bf p})-\m)}-q$ which by positivity should be everywhere positive, with zero possibly allowed on a negligible set w.r.t. Lebesgue measure on the momentum space.
This allows to compute the possible range of the chemical potential, leading to $\m\in(-\infty,\m_q]$, where $\m_q=+\infty$ for $q\in[-1,0]$, and $\m_q=-\ln q$ for $q\in(0,1]$, 0 being the limiting value between the two situations. As it will be more clear below, the the previous cases can be considered as the Fermi/Bose-like alternative, with the separation case $q=0$ is known to be the the Boltzmann (or free) one.

In the Fermi-like/Boltzmann cases $q\in[-1,0]$, including the limiting Boltzmann case $q=0$, we can freely solve \eqref{cedcae} for each $\m\in\br$. So the concept of {\it critical density} plays a role only in the Bose-like case $q\in(0,1]$. For such values of $q$,  \eqref{cedcae} (or equivalently  \eqref{cedcae1}) can be solved only if $\m<\m_q:=-\ln q$, being $\m_q$ the critical value of the chemical potential for which the condensation can occur. Concerning the critical density at the inverse temperature function $\b$, it is defined as
$\r_c^{(q)}(\b):=\r^{(q)}(\b,\m_q)$ which leads to
\begin{equation}
\label{qqdensi}
\r_c^{(q)}(\b):=\int_{\br^d}\frac{\di^d{\bf p}}{e^{\b(h(p))h(p)-\m_q}-q}=\frac{\r_c^{(1)}(\b)}q\,.
\end{equation}
We will se below the role played by the limiting value of the chemical potential $\m_q$ and critical density $\r_c^{(q)}$, in the appearance of the condensate of $q$-particles.
As before, $\r_c^{(1)}$ is simply denoted as $\r_c$. We also omit the dependence on $\b$ when this causes no matter of confusion.

According to the previous computations, the following theorem explains the form of the kernels which are the candidate to reproduce the two-point function of a quasi-free state satisfying LEP. Among the other things,
we show how the condensation regime naturally emerge without using the thermodynamic limit of finite volume theories. 
\begin{Thm}
\label{main}
Suppose that the inverse temperature function $\b$ fulfils (i)-(iii) above. Let $F$ be a kernel on $\br^d$ such that $e^{(\b(h({\bf k}))h({\bf k})-\m)}F({\bf p},{\bf k})$ is still a distribution satisfying \eqref{cedcae1}. Then the following assertions hold true.
\begin{itemize}
\item[(i)] If $\m>-\ln(0\vee q)$ (with the convention that $-\ln0=+\infty$), no of such kernels can exist.
\item[(ii)] {\bf non condensation regime}: For each $\m<-\ln(0\vee q)$, there exists only one kernel as above having the form
\begin{equation}
\label{kdensi}
F({\bf k},{\bf p})=\frac{\d({\bf k}-{\bf p})}{e^{\b(h(p))h(p)-\m}-q}\,.
\end{equation}
\item[(iii)] {\bf condensation regime}: Let $q\in(0,1]$ and $\m=\m_q$. If $E=\emptyset$, then $\r_c^{(q)}<+\infty$ and 
\begin{equation}
\label{kdensi0}
F({\bf p},{\bf k})=\frac{\d({\bf p}-{\bf k})}{q(e^{\b(h({\bf p}))h({\bf p})}-1)}
\end{equation}
is the unique kernel fulfilling the hypotheses.

Let $E\neq\emptyset$, and suppose that $\b\lceil_{\br_+\backslash E}$ is infinitely often differentiable. 

If $\r^{(q)}_c(\b)=+\infty$, there is no of such kernels satisfying the additional condition
$|f|\leq|g|\Rightarrow F(f\otimes\bar f)\leq F(g\otimes\bar g)$.

If $\r^{(q)}_c<+\infty$, then 
$F$ assumes the form
\begin{equation}
\label{kdensi1}
F({\bf p},{\bf k})=\frac{\d({\bf p}-{\bf k})}{q(e^{\b(h({\bf p}))h({\bf p})}-1)}+G({\bf p},{\bf k})\,,
\end{equation}
where $G\in\cd'(\br^d\times\br^d)$ is supported in $\bigcup\big\{\bs_{r^{\frac1s}}\mid r\in E\big\}$. 
\end{itemize}
\end{Thm}
\begin{proof}
(i) and (ii): Being $n_{\bf p}=\frac{1}{e^{\b(h({\bf p}))h({\bf p})-\m}-q}$
the density of occupation number at momentum $\bf p$ and chemical potential $\m$, it must be a.e. positive, w.r.t. the Lebesgue  measure. This excludes $\m>-\ln(0\vee q)$. Conversely, if $\m<-\ln(0\vee q)$ we can uniquely solve \eqref{cedcae} by obtaining \eqref{kdensi}.
\vskip.3cm
\noindent
(iii): By (ii), the condensation regime can occur only if $q\in(0,1]$ and $\m=-\ln q$. So we reduce the matter to this case. If $E=\emptyset$, then condition (iii) at the beginning of the present section together \eqref{qqdensi} imply that $\r_c^{(q)}<+\infty$. In addition, being $e^{\b(h({\bf p}))h({\bf p})}-1>0$ everywhere, \eqref{cedcae1} can be freely solved w.r.t. $F$ giving as the unique solution \eqref{kdensi0}. Suppose now $E\neq\emptyset$. We first show that, under the quite natural conditions 
imposed to $\b$ and to the kernel $F$, the critical density must be finite. Thanks to \eqref{qqdensi}, we can reduce the matter to $q=1$. Suppose that $\r_c=+\infty$. By condition (iii) above, it means that $\int_{B_{x_0}}\frac{x^{\frac{d}{s}-1}}{e^{\b(x)x}-1}\di x=+\infty$, at least for some small neighbourhood $B_{x_0}$ centred in $x_0\in E$ of radius $\d$ containing eventually no further points of 
$E$. Choose a positive function $\f\in\cd(\br^d)$ which is identically 1 on the spherical shell of thickness $\d$ around the sphere 
$h({\bf p})=x_0$ (the sphere of radius $\d/2$ centred in ${\bf 0}$ if $x_0=0$), and whose support does not contain further points ${\bf p}$ with $h({\bf p})\in E$. Choose a sequence $\{\f_n\}_{n\in\bn}\subset\cd(\br^d)$ of positive functions with $\f_n<\f$, such that the sphere $h({\bf p})=x_0$ is not contained in their support, monotonically converging point-wise a.e. (w.r.t. the Lebesgue measure) to $\f$. By hypothesis on $F$, we have for each $n\in\bn$,
$$
\int_{\br^d\times\br^d}F({\bf p},{\bf k})\f_n({\bf p})\f_n({\bf k})\di^d {\bf p}\di^d {\bf k}
\leq\int_{\br^d\times\br^d}F({\bf p},{\bf k})\f({\bf p})\f({\bf k})\di^d {\bf p}\di^d {\bf k}<+\infty\,.
$$
On the other hand, as $\frac{\f_n({\bf p})\f_n({\bf k})}{e^{\b(h({\bf p}))h({\bf p})}-1}$ is a smooth compactly supported function, 
by using \eqref{cedcae1} and the Monotone Convergence Theorem, we obtain
\begin{align*}
&\int_{\br^d\times\br^d}F({\bf p},{\bf k})\f_n({\bf p})\f_n({\bf k})\di^d {\bf p}\di^d {\bf k}\\
=&\int_{\br^d\times\br^d}\bigg(\big(e^{\b(h({\bf p}))h({\bf p})}-1\big)F({\bf p},{\bf k})\bigg)
\frac{\f_n({\bf p})\f_n({\bf k})}{e^{\b(h({\bf p}))h({\bf p})}-1}\di^d {\bf p}\di^d {\bf k}\\
=&\int_{\br^d\times\br^d}\frac{\d({\bf p}-{\bf k})}{e^{\b(h({\bf p}))h({\bf p})}-1}\f_n({\bf p})\f_n({\bf k})\di^d {\bf p}\di^d {\bf k}\\
=&\int_{\br^d}\frac{\f_n({\bf p})^2}{e^{\b(h({\bf p}))h({\bf p})}-1}\di^d {\bf p}\to
\int_{\br^d}\frac{\f({\bf p})^2}{e^{\b(h({\bf p}))h({\bf p})}-1}\di^d {\bf p}\\
&\geq\int_{B}\frac{\di^d {\bf p}}{e^{\b(h({\bf p}))h({\bf p})}-1}\di^d {\bf p}=+\infty
\end{align*}
which is a contradiction. 

Suppose now that $\r_c<+\infty$. Then $\frac{\d({\bf p}-{\bf k})}{e^{\b(h({\bf p}))h({\bf p})}-1}$ defines a distribution on $\br^d\times\br^d$ because of (iii) at the beginning of the section. Then 
$$
G({\bf p},{\bf k}):=F({\bf p},{\bf k})-\frac{\d({\bf p}-{\bf k})}{e^{\b(h({\bf p}))h({\bf p})}-1}
$$
defines also a distribution. Now we show that $\supp G=\bigcup\big\{\bs_{r^{\frac1s}}\mid r\in E\big\}$. Pick $\f\in\cd(\br^d\times\br^d)$ with 
$K:=\supp\f\subset\big\{\bs_{r^{\frac1s}}\mid r\in E\big\}^c$. We can suppose without loosing generality that $K$ is a smooth manifold with boundary. Consider the compact slices $K_1:=K\cap\{({\bf p},{\bf k})\mid h({\bf k})\in E\}$, $K_2:=K\cap\{({\bf p},{\bf k})\mid h({\bf p})\in F\}$ and choose open neighbourhoods $V_i\supset K_i$, $i=1,2$, such that the distances between $K\backslash(V_1\cup V_2)$ and $K_i$, $i=1,2$, $K\backslash(V_1\cup V_2)$ and $\bigcup\big\{\bs_{r^{\frac1s}}\mid r\in F\big\}$ is greater of $2\d$ for some $\d>0$. Then there exists a finite open covering of $K$ made of $V_i\supset K_i$, $i=1,2$, (which cover the $K_i$), and some balls $\{B_j\}_{j=1}^n$ centred in some points $\{\xi_j\}_{j=1}^n\subset K\backslash(V_1\cup V_2)$ of radius $\d$ (which cover $K\backslash(V_1\cup V_2)$). Fix a smooth partition of the unity 
$$
\psi=\psi_1+\psi_2+\sum_{i=1}^n f_i
$$
subordinate to this covering. After using \eqref{cedcae1} for the addendum containing $\psi_1$, \eqref{cedcae} for that containing $\psi_2$, and finally indifferently one of them for the addenda containing the $f_i$, we first note that 
$$
\frac{\f({\bf p},{\bf k})}{e^{\b(h({\bf p}))h({\bf p})}-1}\bigg(\psi_1({\bf p},{\bf k})+\sum_{i=1}^n f_i({\bf p},{\bf k})\bigg)\,,
\quad\frac{\f({\bf p},{\bf k})\psi_2({\bf p},{\bf k})}{e^{\b(h({\bf k}))h({\bf k})}-1}
$$
are smooth compactly supported functions. Then we compute,
\begin{align*}
&\int_{\br^d\times\br^d}
G({\bf p},{\bf k})
\f({\bf p},{\bf k})\di^d {\bf p}\di^d {\bf k}\\
=&\int_{\br^d\times\br^d}\bigg((e^{\b(h({\bf p}))h({\bf p})}-1)F({\bf p},{\bf k})-\d({\bf p}-{\bf k})\bigg)\\
\times&\frac{\f({\bf p},{\bf k})}{e^{\b(h({\bf p}))h({\bf p})}-1}
\bigg(\psi_1({\bf p},{\bf k})+\sum_{i=1}^n f_i({\bf p},{\bf k})\bigg)
\di^d {\bf p}\di^d {\bf k}\\
+&\int_{\br^d\times\br^d}\bigg((e^{\b(h({\bf k}))h({\bf k})}-1)F({\bf p},{\bf k})-\d({\bf p}-{\bf k})\bigg)\\
\times&\frac{\f({\bf p},{\bf k})\psi_2({\bf p},{\bf k})}{e^{\b(h({\bf k}))h({\bf k})}-1}
\di^d {\bf p}\di^d {\bf k}=0\,.
\end{align*}
\end{proof}
We point out that the investigation of the condensation phenomena of $q$-particles in the setting considered in the present section, can be reduced to find the positive definite solutions of the equation \eqref{cedcae} (or equivalently of the twin one \eqref{cedcae1}) in the space of the distributions. As shown in Theorem \ref{main}, it is a difficult task to find all the possible solutions of that equation. In addition, not all the distributions $G$ appearing in \eqref{kdensi1} give rise to a quasi-free states satisfying LEP and exhibiting condensation. However, we will see in Section \ref{0nw1} that, even in the standard equilibrium case when $\b$ is the inverse temperature of the system, \eqref{cedcae} admits solutions of the form \eqref{kdensi1} with nontrivial $G$, hence describing condensation effects, which are completely new and unexpected. Concerning the quasi-free states described in the previous section, the two-point function \eqref{rotcinv} also have the form \eqref{kdensi1} with nontrivial $G$, including those in Proposition \ref{pprroo} which are rotationally invariant, even if the local inverse temperature is more general than those considered in Theorem \ref{main}.


\section{new states describing condensation}
\label{0nw1}

In the present section we show that new states exhibiting condensation can occur even in the standard equilibrium thermodynamics. We reduce the matter to the Bosonic case with $\b$ the fixed inverse temperature. As usual, we put $h({\bf p})=p^s$, 
$s\geq1$, and reduce the analysis to the simplest case of the partial derivative along only one direction $k_1$ in the momentum space. As for $\f\in\cd(\br^d)$, $e^{\imath p^s t}\f({\bf p})$ is in general not smooth when $s=1$ or $s$ is not integer, we slightly extend the space of the test functions, even if it is not necessary for the unique physical case $s=2$ for which it is possible to exhibit such new states describing BEC. We put 
$$
\gph_0=\svv\{e^{\imath h t}\check\f\mid\f\in\cd(\br^d)\,,t\in\br\}\,.
$$
In this case, the one-particle dynamics is meaningful on $\gph_0$ by construction.
\begin{Prop}
\label{ddeldercs}
Fix $D>0$ and suppose that the critical density $\r_c<+\infty$. Then the states whose two-point function is given by
\begin{equation}
\label{delder}
\om(a^\dagger(\check f)a(\check g))=\int_{\br^d}\frac{f({\bf p})\overline{g({\bf p})}}{e^{\b p^s}-1}\di^d{\bf p}
+D\frac{\partial f}{\partial p_1}({\bf 0})\frac{\partial \bar g}{\partial p_1}({\bf 0})\,,\quad \check f, \check g\in\gph_0\,,
\end{equation}
are equilibrium states for the dynamics generated by the one-particle Hamiltonian $h({\bf p})=p^s$, provided $s>1$.
\end{Prop}
\begin{proof}
Fix a generator $\check f=e^{\imath h t}\check\f$. It is straightforward to see that
$$
\int_{\br^d}\frac{|e^{\b p^s}f({\bf p})|^2}{e^{\b p^s}-1}\di^d{\bf p}<+\infty\,.
$$
In addition, if $s>1$ we get
$$
\frac{\partial(\g f)}{\partial p_1}({\bf p})=\b sp^{s-2}p_1e^{\b p^s}f({\bf p})+e^{\b p^s}\frac{\partial f}{\partial p_1}({\bf p})\,.
$$
Collecting together with $\g({\bf p})=e^{\b p^s}$, we conclude first that $\om\big(a^\dagger(\g f)a(\g f)\big)<+\infty$,
which means that $\g f$ is in the domain of the form describing the two-point function of $\om$. Second,
$\frac{\partial(\g f)}{\partial p_1}({\bf 0})=\frac{\partial f}{\partial p_1}({\bf 0})$,
leading to \eqref{lepccr2}. Namely,  the states in \eqref{delder} are equilibrium states at inverse temperature $\b$ and chemical potential $\m=0$.
\end{proof}
The two-point function of the usual state $\f$ exhibiting the condensation of free Bosons is given by
\begin{equation}
\label{csatzi}
\f(a^\dagger(\check f)a(\check g))=\int_{\br^d}\frac{f({\bf p})\overline{g({\bf p})}}{e^{\b p^s}-1}\di^d{\bf p}
+Df({\bf 0})\overline{g({\bf 0})}\,,
\end{equation}
and has the form \eqref{kdensi1} with
$$
G_\f({\bf p},{\bf k})=\d({\bf p})\d({\bf p}-{\bf k})\,.
$$
Despite the product of two delta distributions, $G_\f({\bf p},{\bf k})$ makes sense as a distribution, giving precisely the condensation portion after the application of Fubini rule in the formal integration as explained below.
Concerning the two-point function in \eqref{delder}, it assumes the form \eqref{kdensi1} with
$$
G_\om({\bf p},{\bf k})=\frac{\partial^2\big(\d({\bf p})\d({\bf p}-{\bf k})\big)}{\partial p_1\partial k_1}
=\frac{\partial^2G_\f({\bf p},{\bf k})}{\partial p_1\partial k_1}\,,
$$
which leads to
\begin{align*}
\int G_\om({\bf p},{\bf k})f({\bf p})\overline{g({\bf k})}&\di^d {\bf p}\di^d {\bf k}
=\int\frac{\partial^2G_\f({\bf p},{\bf k})}{\partial p_1\partial k_1}f({\bf p})\overline{g({\bf k})}\di^d {\bf p}\di^d {\bf k}\\
=&\int\di^d {\bf p}\frac{\partial f({\bf p})}{\partial p_1}\d({\bf p})
\int\di^d {\bf k}\frac{\partial\overline{g({\bf k})}}{\partial k_1}\d({\bf p}-{\bf k})\\
=&\int\di^d {\bf p}\frac{\partial f({\bf p})}{\partial p_1}\frac{\partial\overline{g({\bf p})}}{\partial p_1}\d({\bf p})
=\frac{\partial f}{\partial p_1}({\bf 0})\frac{\partial \bar g}{\partial p_1}({\bf 0})\,.
\end{align*}
The higher (even) derivatives of the delta-distribution can give rise to a condensate distribution only if the power $s$ in $h({\bf p})=p^s$ is sufficiently big. Conversely, it never appear in the case of photon/phonon one-particle Hamiltonian $h({\bf p})=p$, and only states like \eqref{delder} involving the second derivative of the delta-distribution and not the higher ones can appear in the free massive Bosons for which $h({\bf p})=p^2$.

The local density of particles $\r_\om({\bf x})$ of the states in \eqref{delder} is finite everywhere in the configuration space:
\begin{equation}
\label{parab}
\r_\om({\bf x}):=\om(a^\dagger(\d_{{\bf x}})a(\d_{{\bf x}}))=\int_{\br^d}\frac{\di^d{\bf p}}{e^{\b p^s}-1}
+D\bigg|\frac{\partial\widehat{\d_{{\bf x}}}}{\partial p_1}({\bf 0})\bigg|^2
=\int_{\br^d}\frac{\di^d{\bf p}}{e^{\b p^s}-1}
+Dx_1^2\,.
\end{equation}
Conversely, the mean density is infinite due to the contribution of the condensation term
$D\big|\frac{\partial\widehat{\d_{{\bf x}}}}{\partial p_1}({\bf 0})\big|^2=Dx_1^2$. As the standard procedure in constructing
states exhibiting BEC as infinite volume limits of finite volume theory is to fix the mean density of the model, the last computation easily explains because such terms connected to the second derivative cannot appear in the standard investigation based on Bose-Gibbs prescription for the grand canonical ensemble construction.

States which are rotationally invariant (i.e. isotropic) are easily given by
\begin{equation}
\label{delder1}
\om\big(a^\dagger(\check f)a(\check g)\big)=\int_{\br^d}\frac{f({\bf p})\overline{g({\bf p})}}{e^{\b p^s}-1}\di^d{\bf p}
+D\nabla_{\bf p} f({\bf 0}){\bf\cdot}\nabla_{\bf p}\overline{g({\bf 0})}\,,
\end{equation}
as well as states exhibiting condensation, which are connected to higher derivative of the Dirac distribution can be easily wrote down, provided that $s$ is sufficiently big. It is also possible to exhibit such states involving higher derivatives of the delta-distribution in local equilibrium, obtaining a richer situation.

For all the states described in \eqref{delder1} (and also for those in \eqref{delder}), the action of the spatial translations in the configuration space $T_{\bf x}$ is spontaneously broken. We easily compute
\begin{align*}
&\om\circ T_{\bf x}\big(a^\dagger(\check f)a(\check g)\big)=
\om\big(T_{\bf x}(a^\dagger(\check f)a(\check g))\big)\\
=\om\big(a^\dagger(\check f)a(\check g)\big)
+&D\big[x^2f({\bf 0})\overline{g({\bf 0})}+\imath{\bf x}{\bf\cdot}\big(\overline{g({\bf 0})}\nabla_{\bf p} f({\bf 0})-f({\bf 0})\nabla_{\bf p}\overline{g({\bf 0})}\big)\big]\,,
\end{align*}
where as usual the dot denotes the real inner product. Contrarily to the states in \eqref{1ab1} for the rotation symmetry, for the states in \eqref{delder1} also the local density of the particle changes for a term relative to the condensate which is simply computed as
$$
\r_{\om\circ T_{\bf y}}({\bf x})=\r_{\om}({\bf x})+Dy^2\,.
$$
Fix any mollifier $\d_\eps({\bf p})$, converging in the sense of distribution to the delta distribution $\d({\bf p})$ as $\eps\downarrow0$. In order to have an idea of the distribution of the condensate in momentum space of the usual situation \eqref{csatzi}
\begin{figure}[ht]
     \centering
     \psfig{file=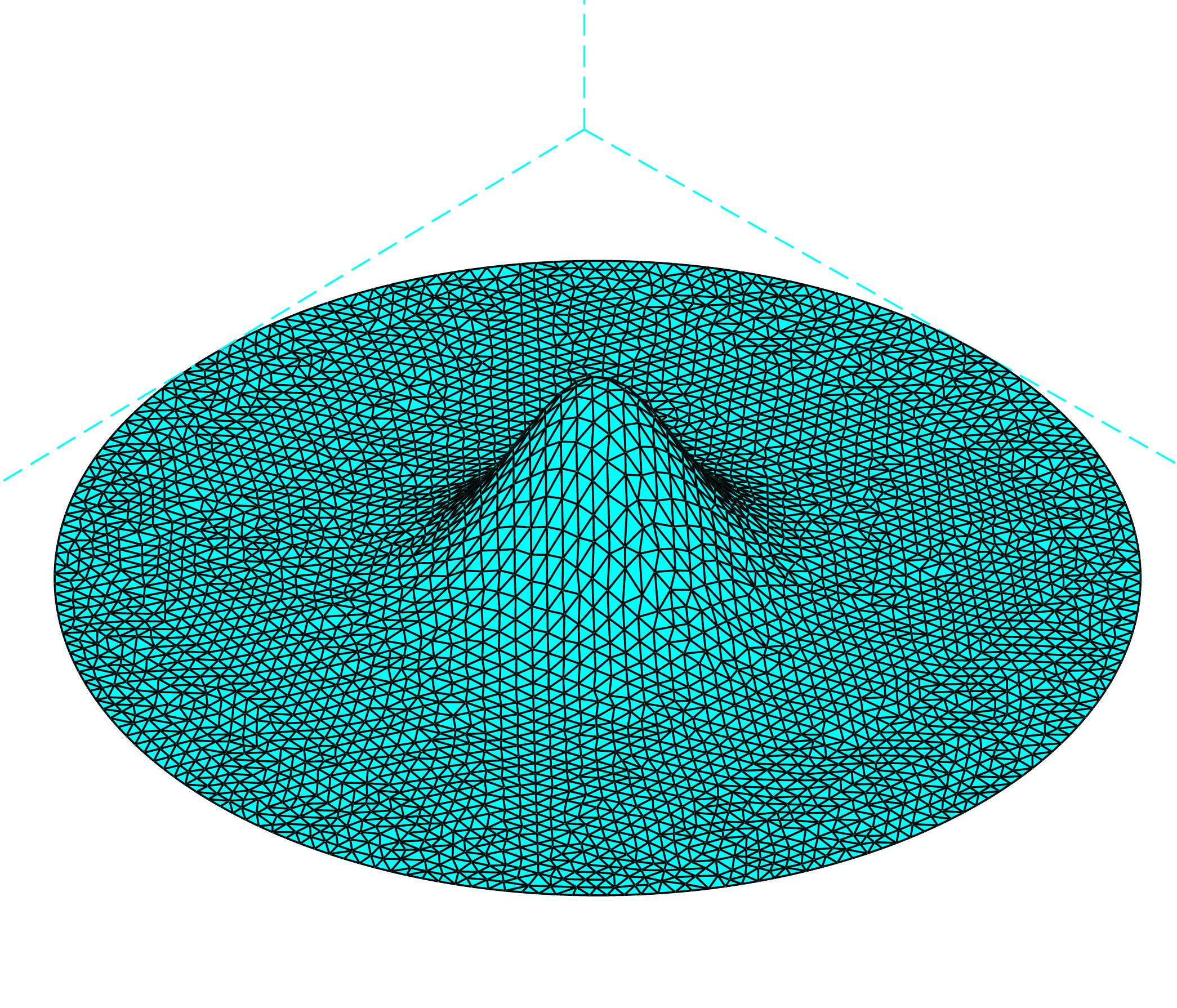,height=2.1in} \,\, \psfig{file=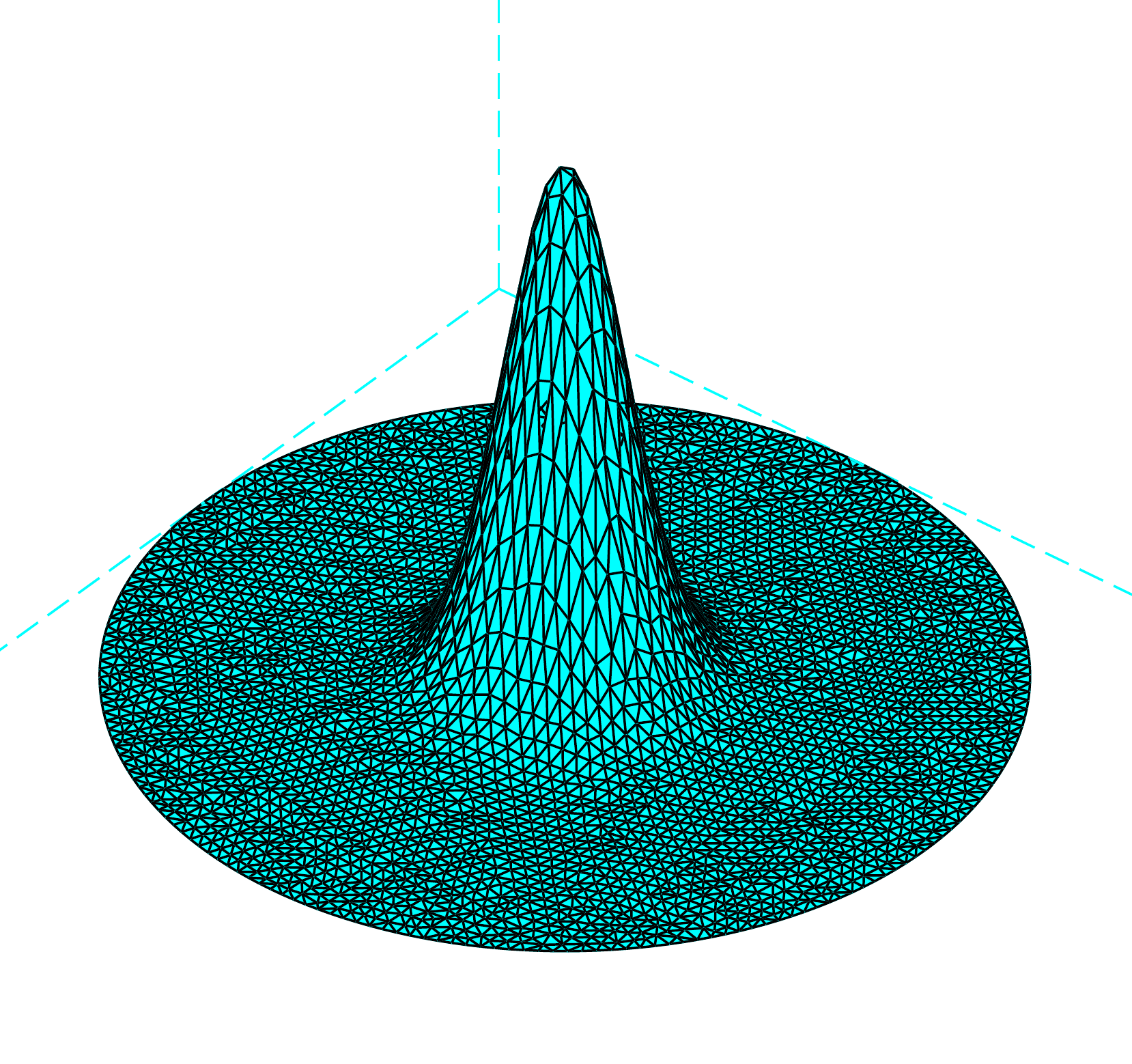,height=2.0in}
     \caption{The density $\r_\f^{(\text{cond})}(k_1,k_2)$ in momentum space relative to the states \eqref{csatzi} for two values $\eps_1>\eps_2$.}
     \label{f1}
     \end{figure}
and that of \eqref{delder1}, we end the present section by drawing the projection in the $k_1$-$k_2$ plane of
\begin{align*}
\r_\f^{(\text{cond})}(k_1,k_2)&=\d_\eps(k_1,k_2,0)^2\,,\quad\\
\r_\om^{(\text{cond})}(k_1,k_2)&=\big\|\nabla_{\bf p}\d_\eps({\bf p}-{\bf k})\lceil_{{\bf p}=0,k_3=0}\big\|^2\,,
\end{align*}
\begin{figure}[ht]
     \centering
     \psfig{file=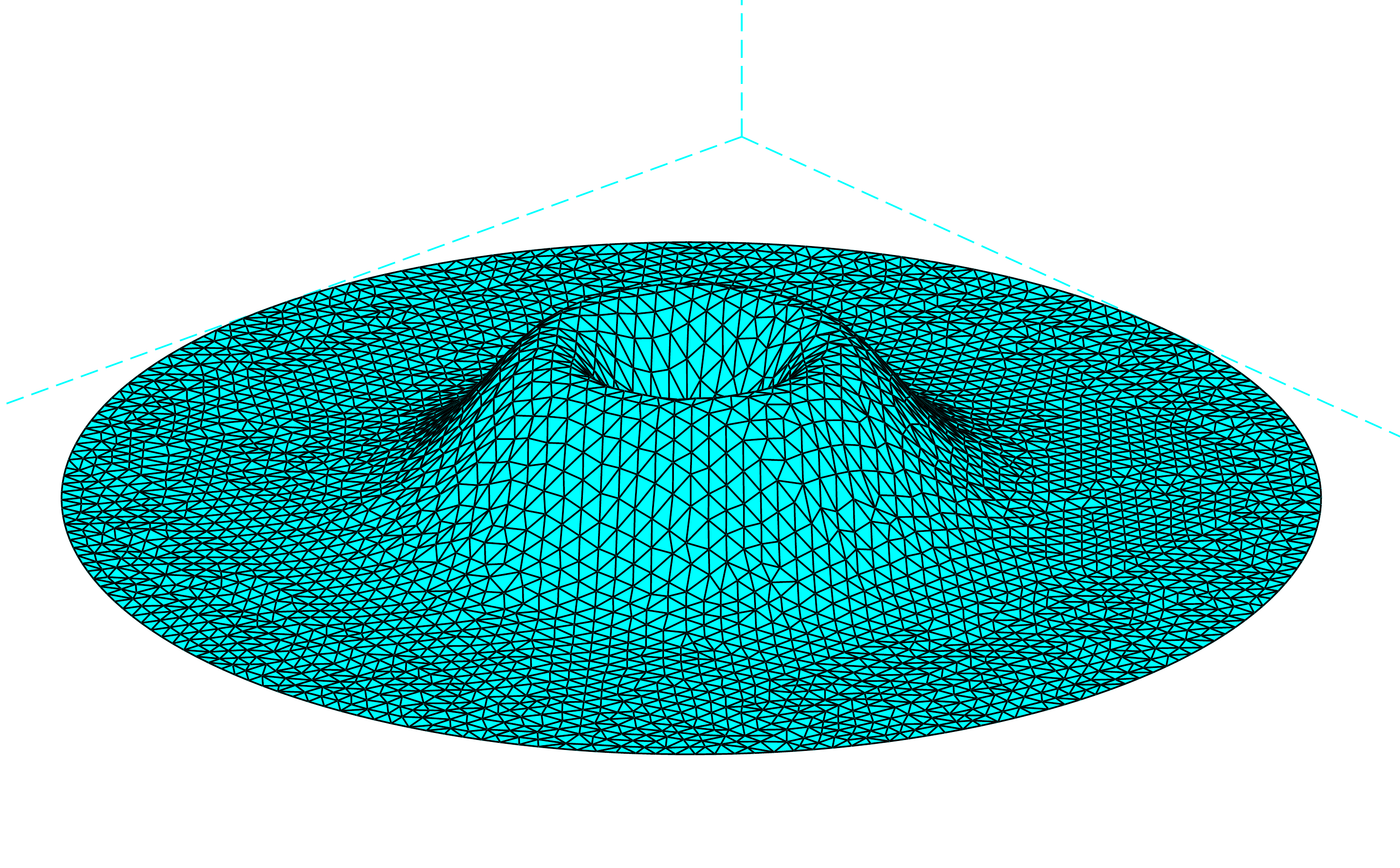,height=1.5in} \,\, \psfig{file=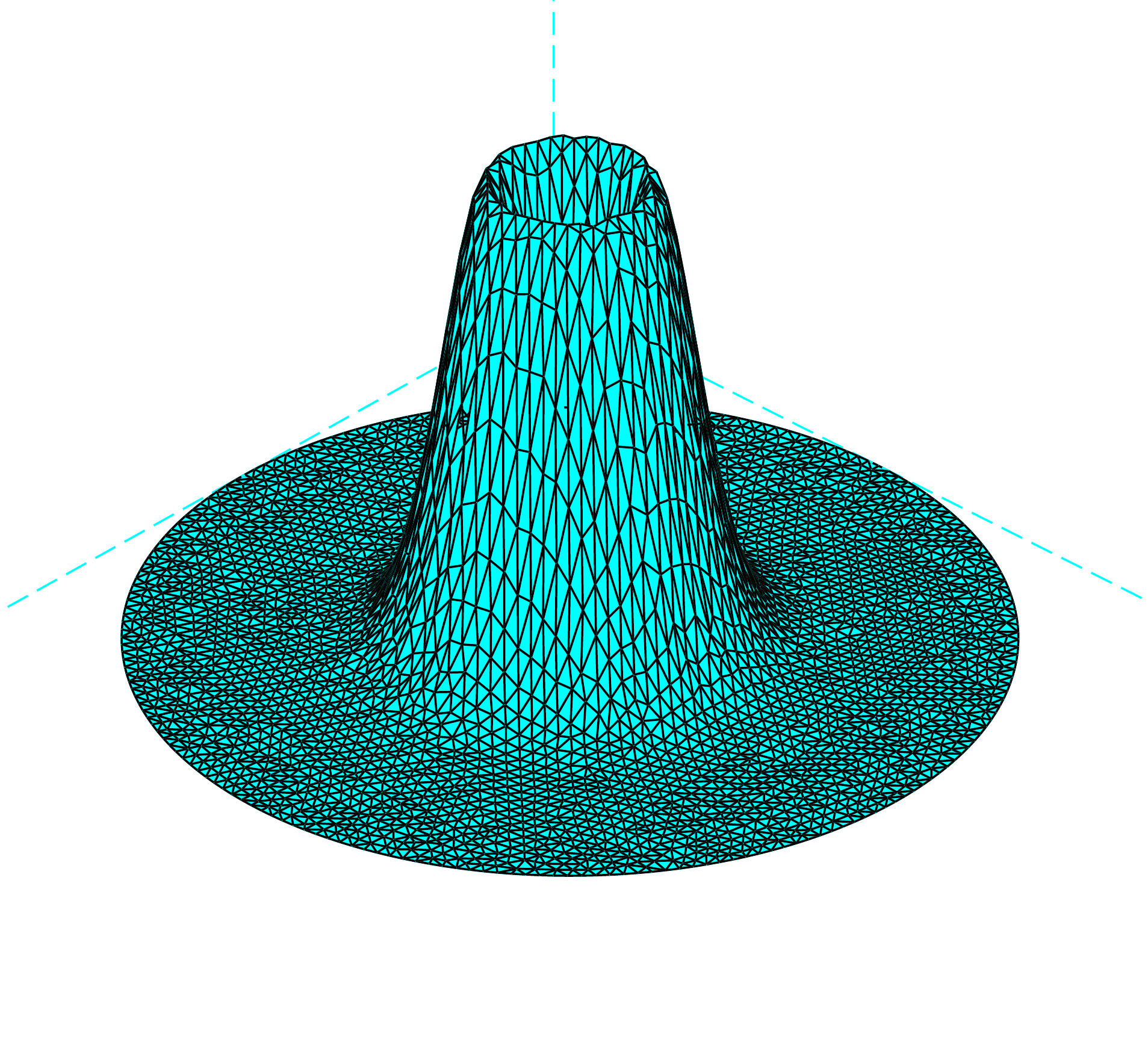,height=1.9in}
     \caption{The density $\r_\om^{(\text{cond})}(k_1,k_2)$ in momentum space relative to the states \eqref{delder1} for two values $\eps_1>\eps_2$.}
     \label{f2}
     \end{figure}
where $\om$ is the standard equilibrium state exhibiting condensation \eqref{csatzi} (cf. Fig. \ref{f1}), and one in \eqref{delder1} (cf. Fig. \ref{f2}), respectively.

\section{remarks and outlook}
\label{outlook}

In order to study the condensation regime, we have met the problem concerning the way to introduce the chemical potential. We briefly discuss a further reasonable choice
\begin{equation}
\label{koch}
n_\eps=\frac1{z^{-1}e^{\b(\eps)\eps}-1}=\frac1{e^{\b(\eps)(\eps-\m)}-1}
\end{equation}
corresponding to choice for the the activity
$$
z(\b,\m)=e^{\b\m}=e^{\b(\eps)\m}\,.
$$
The analogous equation to \eqref{cedcae1} assumes the form
\begin{equation}
\label{cedcaea}
\big(e^{\b(h({\bf p}))(h({\bf p})-\m)}-q\big)F_\om({\bf p},{\bf k})=\d({\bf p}-{\bf k})\,.
\end{equation}
It is simply to see that this different possibility introduces no change for the cases $q\in[-1,0]$, and for the limiting Bosonic case $q=1$. For the Boson-like cases $q\in(0,1)$, 
we compute for the sake of completeness, the critical value of the chemical potential $\m_q$, and hence the relative critical density $\r_c^{(q)}$. The key-point is the following
\begin{Lemma}
\label{claezmzao}
Let the positive function $\b:(0,+\infty)\to[0,+\infty)$ be continuous and vanishing of a set $E\subset(0,+\infty)$ such that the set ${\rm L}(E)$ of its cluster points is discrete.
If $\limsup_{x\downarrow0}<+\infty$, then the function
$G:\br_+\to\br_+$ given by $G(\m):=\sup_{(0,\m]}\big(\b(x)(\m-x)\big)$ is strictly increasing and continuous.
\end{Lemma}
\begin{proof}
Define $\b(0):=\limsup_{x\downarrow0}$ and fix $\m>\l$. There exists a $\bar x\in[0,\l]$ for which $\b(\bar x)>0$ such that by applying Weierstrass Theorem, we get
\begin{align*}
&G(\m)=\sup_{(0,\m]}\big(\b(x)(\m-x)\big)\geq\sup_{(0,\l]}\big(\b(x)(\m-x)\big)\\
\geq&\sup_{(0,\l]}\big(\b(x)(\l-x)\big)+\b(\bar x)(\m-\l)
>\sup_{(0,\l]}\big(\b(x)(\l-x)\big)=G(\l)\,.
\end{align*}
Thus $G(\m)$ is strictly increasing. 
Concerning the continuity, first notice that 
$$
\lim_{\m\downarrow\l}\sup_{(0,\l]}\big(\b(x)(\m-x)\big)=\sup_{(0,\l]}\big(\b(x)(\l-x)\big)\,,
$$
because for $x\leq\l$, 
$$
b(x)(\m-x)=\b(x)(\l-x)+\b(x)(\m-\l)\,.
$$ 
In addition, 
$$
\lim_{\m\downarrow\l}\sup_{[\l,\m]}\big(\b(x)(\m-x)\big)=0\,.
$$
Collecting together, we get
$$
\lim_{\m\downarrow\l}G(\m)=\lim_{\m\downarrow\l}
\big(\sup_{(0,\l]}\big(\b(x)(\m-x)\big)\vee\sup_{[\l,\m]}\big(\b(x)(\m-x)\big)\big)
=G(\l)
$$
For the reverse limit, first notice that, with $\m$ fixed, and $\l\leq\m$, 
$$
\l\mapsto\sup_{(0,\l]}\big(\b(x)(\m-x)\big)
$$
is continuous. In addition, 
$$
\lim_{\l\uparrow\m}\sup_{(0,\l]}\big(\b(x)(\m-\l)\big)=0\,.
$$
Then we get with $\m>0$ fixed, $\l<\m$, and $x\in(0,\l]$,
$$
\b(x)(\m-x)-\sup_{(0,\l]}\big(\b(x)(\m-\l)\big)\leq\b(x)(\l-x)\leq\b(x)(\m-x)\,,
$$
which leads to $\lim_{\l\uparrow\m}G(\l)=G(\m)$.
\end{proof}
The choice \eqref{cedcaea} for the occupation numbers introduces technical troubles only for the unphysical cases $q\in(0,1)$, for which we are going to determine the critical chemical potentials responsible of the condensation effects.
\begin{Prop}
\label{1:cp}
Let $\b$ satisfy the hypotheses of Lemma \ref{claezmzao}.
\begin{itemize}
\item[(i)] Suppose that $\limsup_{x\downarrow0}\b(x)=+\infty$, and define $\m_q=0$.
\item[(ii)] Suppose that $\limsup_{x\downarrow0}\b(x)<+\infty$. Then for $\m>0$, the function
\begin{equation}
\label{qkrict0}
F_q(\m):=\frac{-\ln q}{\sup_{(0,\m]}\big(\b(x)(\m-x)\big)}\,,\quad q\in(0,1)\,,
\end{equation}
is continuous and strictly decreasing. Define $\m_q=0$ if
the equation
\begin{equation}
\label{qkrict}
F_q(\m)=1\,,
\end{equation}
has no solution which happens if and only if $\limsup_{x\downarrow0}\b(x)\leq1$, or $\m_q$ as the unique solution of \eqref{qkrict}. 
\end{itemize}
In both situations, the following hold true.
If $\m<\m_q$ then
\begin{equation*}
\big(e^{\b(x)(x-\m)}-q\big)\geq\d>0\,,\quad x\in\br_+\,,
\end{equation*}
and if $\m>\m_q$ there exists an open interval $I\subset\br_+$ such that
\begin{equation*}
\big(e^{\b(x)(x-\m)}-q\big)\leq-\d<0\,,\quad x\in I\,.
\end{equation*}
\end{Prop}
\begin{proof}
If $\limsup_{x\downarrow0}\b(x)<+\infty$, Lemma \ref{claezmzao} assures that \eqref{qkrict0} is a continuous strictly decreasing function.
In all the situations, if $x\geq\m$, then $e^{\b(x)(x-\m)}\geq1>q$. If $\m<\m_q$ and $x<\m$ then for some $\eps>0$,
$$
\b(x)(\m-x)\leq\b(x)(\m_q-x)-\eps
\leq\sup_{(0,\m_q]}\big(\b(x)(\m-x)\big)-\eps=-\ln q-\eps\,,
$$
which leads to
$$
e^{\b(x)(x-\m)}\geq qe^{\eps}>q\,.
$$
Conversely, if $\m>\m_q$, we get in all the situations that there exists an open interval $I\in(0,\m)$ such that for some $\eps>0$,
$$
\b(x)(\m-x)\geq-\ln q +\eps\,,\quad x\in I\,.
$$
But this immediately implies for $x\in I$,
$$
e^{\b(x)(x-\m)}\leq qe^{-\eps}<q\,.
$$
\end{proof}
Notice that, being $\ln q=0$, \eqref{qkrict} is meaningful even for $q=1$ providing $\m_1=0$ as expected. In addition, if $\b$ is constant, then $\b\m_q=-\ln q$, which reduces to the usual formula of equilibrium thermodynamics. Proposition \ref{1:cp} explains why $\m_q$ solving \eqref{qkrict} is the critical value of the chemical potential for the $q$-deformed situation, 
$q\in(0,1)$, for the choice \eqref{koch} concerning the introduction of the chemical potential.
The corresponding critical density is easily given by
$$
\r_c^{(q)}=\int_{\br^d}\frac{\di^d{\bf p}}{e^{\b(h({\bf p}))(h({\bf p})-\m_q)}-q}\,.
$$
Unfortunately, there is no direct relation between $\{\r_c^{(q)}\mid q\in(0,1)\}$ and $\r_c^{(1)}$. This introduces some additional technical troubles in order to manage \eqref{cedcaea} for the unphysical cases $0<q<1$ in the condensation regime $\m=\m_q$, without affecting the substance of the results. We do not pursue this direction, which perhaps can be managed in details for particular choices of the local temperature function $\b$.

We have discussed new ideas of Local Equilibrium Principle which is a suitable candidate to select stationary states in non equilibrium thermodynamics. Roughly speaking, the Local Equilibrium Principle asserts that the temperature is a function of the energy levels. Even if the Local Equilibrium Principle cannot be directly extended to general dynamical systems describing models with infinitely many degrees of freedom, as well as general condensed matter states appearing in nature, 
we have provided a wide class of nontrivial new and unexpected examples involving Boson particles and describing the Bose-Einstein Condensation. It is expected that the standard construction based on the Bose-Gibbs grand canonical ensemble involving the thermodynamical limit of the finite volume theories cannot directly cover all the examples we have found. This suggested a new approach based on the theory of the Distributions. This new approach allowed us to manage in an unified way the $q$-Commutation Relations, and the condensation of the so called $q$-particles can appear in a natural way only for the Bose-like situation $q\in(0,1]$. 
By using the new approach to BEC described in Section \ref{new}, we found new examples even in the standard case of equilibrium thermodynamics when the temperature is kept fixed by an external thermal bath. All these new states might find applications in high density/temperature non equilibrium physics, quantum optics, theory of superconductivity. Among such possible applications, we briefly outline the following ones as promising perspectives. 

The first one concerns the states in the usual case of equilibrium thermodynamics described in Section \ref{0nw1}.
One consider the axially symmetric states whose two-point function is given by
\begin{equation*}
\om(a^\dagger(\check f)a(\check g))=\int_{\br^d}\frac{f({\bf p})\overline{g({\bf p})}}{e^{\b p^2/2m^*}-1}\di^3{\bf p}
+D\bigg(\frac{\partial f}{\partial p_x}({\bf 0})\frac{\partial \bar g}{\partial p_x}({\bf 0})
+\frac{\partial f}{\partial p_y}({\bf 0})\frac{\partial \bar g}{\partial p_y}({\bf 0})\bigg)
\end{equation*}
concerning the distribution of some quasi-particle of effective mass $m^*$ at inverse temperature $\b>0$. The spatial local density is easily computed as in \eqref{parab}
\begin{equation}
\label{meisc}
\r_\om(x,y,z)=\int_{\br^d}\frac{\di^3{\bf p}}{e^{\b p^2/2m^*}-1}
+D\big(x^2+y^2)\,.
\end{equation}
The last addendum 
$$
\r^\text{cond}_\om(x,y,z)=D\big(x^2+y^2)
$$
describes the spatial portion of the condensate, and has the form of a $z$-axial symmetric paraboloid. The helium superfluidity was explained (cf. \cite{L}) by the excitations of the quasi-particles corresponding to the rotonic part of the spectrum of $He_4$  as briefly outlined in the introduction. Even if the superfluidity seems not directly connected with the BEC, it should be noted the surprising analogy between the spatial distribution of the condensate \eqref{meisc} and the profile (i.e. the meniscus) of the rotating $He_4$ superfluid, apart from the finer regular structure made of vortices, appearing inside the rotating superfluid, see e.g. \cite{W}. 

The second one concerns the cosmological application of the Local Equilibrium Principle for highly concentrated systems. For example, due to the net energy flow, the black hole evaporation (cf. \cite{B1, B2, Ha, Vi}) is a non equilibrium process. Then we can find states which are far from the equilibrium, for which the temperature is a function of the total energy of the black hole which of course changes during the evaporation process. Another possible application might be to the thermodynamics of neutron, and also of exotic Boson and quark stars. For such systems subjected to extreme conditions, first it might be expected that NESS are the natural candidates to describe their thermodynamics. Second, due to high pressure, the Fermions might form BCS pairs like in $He_3$ or in superconductors, and the BEC might take place.
On the other hand, the Local Equilibrium Principle allows states admitting portions of the condensate even in excited levels. In fact, states similar to those in \eqref{rotcinv} with 
$$
h({\bf p})=c\sqrt{p^2+m^2c^2}-mc^2
$$
the one-particle Hamiltonian of a free Boson gas in a neutron or exotic star (with $m$ being the mass of the involved Boson and $c$ the speed of light, where in a rough approach we are neglecting the gravitational interaction between particles), describe a portion of the condensate in excited levels which cannot appear in the usual equilibrium thermodynamics. Such possible condensation effects in excited levels might partially explain the open problem of the dark matter of the universe. 

We conclude by pointing out that the last two possible applications of the Local Equilibrium Principle are only ideas which deserve of further insights, and are still very far to be understood at this stage.

\section*{Acknowledgement} 

The second-named author was partially supported by Italian INDAM--GNAMPA. He kindly acknowledges M. Cirillo, and L. Tomassini, S. Viaggiu for fruitful discussions concerning the possible applications of the presented results to superfluidity, and to the cosmology, respectively.

\end{document}